\documentclass[sigconf]{acmart}

\usepackage{booktabs} 
\usepackage{graphicx}
\usepackage[show]{chato-notes}
\usepackage{amsmath,amssymb,amsthm}
\usepackage{epsfig}
\usepackage{algorithm}
\usepackage{algorithmic}
\usepackage{url}
\usepackage{hyperref}
\usepackage{xspace}
\usepackage{tikz}
\usepackage{flushend}

\usetikzlibrary{shadows}
\usetikzlibrary{positioning}
\usepackage{xcolor}
\usepackage[framemethod=TikZ]{mdframed}
\usepackage{tkz-berge,float}
\makeatletter
\newsavebox\myboxA
\newsavebox\myboxB
\newlength\mylenA

\newcommand*\xoverline[2][0.75]{%
    \sbox{\myboxA}{$\m@th#2$}%
    \setbox\myboxB\null
    \ht\myboxB=\ht\myboxA%
    \dp\myboxB=\dp\myboxA%
    \wd\myboxB=#1\wd\myboxA
    \sbox\myboxB{$\m@th\overline{\copy\myboxB}$}
    \setlength\mylenA{\the\wd\myboxA}
    \addtolength\mylenA{-\the\wd\myboxB}%
    \ifdim\wd\myboxB<\wd\myboxA%
       \rlap{\hskip 0.5\mylenA\usebox\myboxB}{\usebox\myboxA}%
    \else
        \hskip -0.5\mylenA\rlap{\usebox\myboxA}{\hskip 0.5\mylenA\usebox\myboxB}%
    \fi}
\makeatother

\hypersetup{
    bookmarks=true,         
    unicode=false,          
    pdftoolbar=true,        
    pdfmenubar=true,        
    pdffitwindow=false,     
    pdfstartview={FitH},    
    pdftitle={My title},    
    pdfauthor={Author},     
    pdfsubject={Subject},   
    pdfcreator={Creator},   
    pdfproducer={Producer}, 
    pdfnewwindow=true,      
    colorlinks=true,       
    linkcolor=black,          
    citecolor=black,        
    filecolor=black,      
    urlcolor=black           
}

\newcommand{\spara}[1]{\smallskip\noindent{\bf #1}}

\newtheorem{problem}{Problem}

\DeclareMathOperator*{\argmin}{arg\,min}

\newcommand{\bigO}{\ensuremath{\mathcal{O}}\xspace}

\newcommand{\lap}{\ensuremath{L}}
\newcommand{\lapi}{\ensuremath{L^{(i)}}}
\newcommand{\lapR}{\ensuremath{L^{R}}}
\newcommand{\posE}{\ensuremath{E^+}}
\newcommand{\negE}{\ensuremath{E^-}}
\newcommand{\graph}{\ensuremath{G}}
\newcommand{\rankvec}{\ensuremath{\vec{r}}}

\newcommand{\mbs}{{\small MBS}\xspace}
\newcommand{\mbse}{{\small MBS-EDGE}\xspace}
\newcommand{\fpt}{{\small FPT}\xspace}
\newcommand{\sdp}{{\small SDP}\xspace}

\newcommand{\tribes}{{\sc HighlandTribes}}
\newcommand{\cloister}{{\sc Cloister}}
\newcommand{\congress}{{\sc Congress}}
\newcommand{\bitcoin}{{\sc Bitcoin}}
\newcommand{\referendum}{{\sc TwitterReferendum}}
\newcommand{\election}{{\sc WikiElections}}
\newcommand{\slashdot}{{\sc Slashdot}}
\newcommand{\epinions}{{\sc Epinions}}
\newcommand{\conflict}{{\sc WikiConflict}}
\newcommand{\politics}{{\sc WikiPolitics}}

\newcommand{\ouralgo}{{\sc Timbal}\xspace}
\newcommand{\randprepro}{{\sc Subsample}\xspace}

\newcommand{\ggmz}{{\sc Ggmz}\xspace}
\newcommand{\grasp}{{\sc Grasp}\xspace}

\newcommand{\eige}{{\sc Eigen}\xspace}

\newcommand{\NP}{\ensuremath{\mathbf{NP}}\xspace}
\newcommand{\NPhard}{{\NP-hard}\xspace}

\renewcommand{\vec}[1]{\mathbf{#1}}
\newcommand{\sign}[1]{ \text{sign}(#1)}
\newcommand{\tridist}{1cm}

\newcommand{\squishlist}{
 \begin{list}{$\bullet$}
  {  \setlength{\itemsep}{0pt}
     \setlength{\parsep}{1pt}
     \setlength{\topsep}{1pt}
     \setlength{\partopsep}{0pt}
     \setlength{\leftmargin}{1.5em}
     \setlength{\labelwidth}{1em}
     \setlength{\labelsep}{0.5em}
} }
\newcommand{\squishlisttight}{
 \begin{list}{$\bullet$}
  { \setlength{\itemsep}{0pt}
    \setlength{\parsep}{0pt}
    \setlength{\topsep}{0pt}
    \setlength{\partopsep}{0pt}
    \setlength{\leftmargin}{2em}
    \setlength{\labelwidth}{1.5em}
    \setlength{\labelsep}{0.5em}
} }

\newcommand{\squishdesc}{
 \begin{list}{}
  {  \setlength{\itemsep}{0pt}
     \setlength{\parsep}{3pt}
     \setlength{\topsep}{3pt}
     \setlength{\partopsep}{0pt}
     \setlength{\leftmargin}{1em}
     \setlength{\labelwidth}{1.5em}
     \setlength{\labelsep}{0.5em}
} }
\newcommand{\squishend}{
  \end{list}
}

\usepackage{amsmath,amssymb,amsfonts}
\usepackage{algorithmic}
\usepackage{graphicx}
\usepackage{textcomp}
\usepackage{xcolor}

\graphicspath{{.}{images/}}

\AtBeginDocument{%
  \providecommand\BibTeX{{%
    \normalfont B\kern-0.5em{\scshape i\kern-0.25em b}\kern-0.8em\TeX}}}

\setcopyright{iw3c2w3}
\copyrightyear{2020}
\acmYear{2020}
\acmDOI{10.1145/3366423.3380212}

\acmConference[WWW '20]{Proceedings of The Web Conference 2020}{April 20--24, 2020}{Taipei, Taiwan}
\acmBooktitle{Proceedings of The Web Conference 2020 (WWW '20), April 20--24, 2020, Taipei, Taiwan} 
\acmPrice{}
\acmISBN{978-1-4503-7023-3/20/04}

\begin{document}

\title{Finding large balanced subgraphs in signed networks}

\author{Bruno Ordozgoiti}
\affiliation{\institution{Aalto University}}
\email{bruno.ordozgoiti@aalto.fi}

\author{Antonis Matakos}
\affiliation{\institution{Aalto University}}
\email{antonis.matakos@aalto.fi}

\author{Aristides Gionis}
\authornote{This work was done while the author was with Aalto University.}
\affiliation{\institution{KTH Royal Institute of Technology}}
\email{argioni@kth.se}

\renewcommand{\shortauthors}{Ordozgoiti, Matakos \& Gionis.}

\begin{abstract}
Signed networks are graphs whose edges are labelled with either a positive or a negative sign, and can be used to capture nuances in interactions that are missed by their unsigned counterparts.
The concept of balance in signed graph theory determines whether a network can be partitioned into two perfectly opposing subsets, and is therefore useful for modelling phenomena such as the existence of polarized communities in social networks. While determining whether a graph is balanced is easy, finding a large balanced subgraph is hard. The few heuristics available in the literature for this purpose are either ineffective or non-scalable. In this paper we propose an efficient algorithm for finding large balanced subgraphs in signed networks. The algorithm relies on signed spectral theory and a novel bound for perturbations of the graph Laplacian. In a wide variety of experiments on real-world data we show that our algorithm can find balanced subgraphs much larger than those detected by existing methods, and in addition, it is faster. We test its scalability on graphs of up to 34 million edges.
\end{abstract}

\ccsdesc[500]{Theory of computation~Graph algorithms analysis}

\keywords{graph mining, signed graphs, dense subgraph, community detection}

\maketitle

\section{Introduction}

Social-media platforms have taken hold as one of the main forms of communication in today's society. Despite having served to facilitate connections between individuals, in recent years we have observed an array of negative phenomena associated to these technologies. Among other, these platforms seem to contribute to the polarization of political deliberation, which can be detrimental to the health of democracy. Thus, the study of methods to detect and mitigate polarization in online debates is becoming an increasingly compelling topic 
\cite{mejova2014controversy,liao2014can,vydiswaran2015overcoming,morales2015measuring,garimella2017reducing,garimella2018quantifying}.

Many social-media platforms can be represented by graphs. Thus, graph theory has found a variety of applications in this domain over the last few decades, such as community detection \cite{fortunato2010community}, partitioning \cite{arora2009expander}, and recommendation \cite{page1999pagerank}. One limitation of the graph representations usually employed in the literature is that they can capture the existence, or even the strength, of connections between vertices, 
but not their disposition. For instance, in a social network, vertices may represent people and edges interactions between them. By relying just on this information we cannot know whether each interaction is friendly or hostile. 


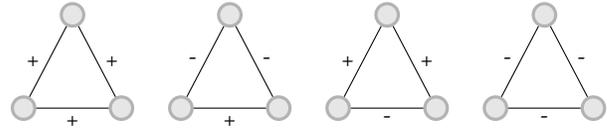
\begin{figure}[t]
  \begin{center}    
    \begin{tikzpicture}[
        roundnode/.style={circle, draw=gray!60, fill=gray!20, very thick, minimum size=7mm},
      ]
      \node[roundnode, minimum size=.2cm]      (a1)                              {};
      \node[roundnode, minimum size=.2cm]      (a2)       [below left = \tridist and 0.4\tridist of a1] {};
      \node[roundnode, minimum size=.2cm]      (a3)       [below right = \tridist and 0.4\tridist of a1] {};

      \draw (a1) -- (a2) node[left, pos=0.5]{+};
      \draw (a1) -- (a3) node[right, pos=0.5]{+};
      \draw (a2) -- (a3) node[below, pos=0.5]{+};

      \node[roundnode, minimum size=.2cm]      (b1)       [right=1.75\tridist of a1]                      {};
      \node[roundnode, minimum size=.2cm]      (b2)       [below left = \tridist and 0.4\tridist of b1] {};
      \node[roundnode, minimum size=.2cm]      (b3)       [below right = \tridist and 0.4\tridist of b1] {};

      \draw (b1) -- (b2) node[left, pos=0.5]{-};
      \draw (b1) -- (b3) node[right, pos=0.5]{-};
      \draw (b2) -- (b3) node[below, pos=0.5]{+};

      \node[roundnode, minimum size=.2cm]      (c1)       [right=1.75\tridist of b1]                      {};
      \node[roundnode, minimum size=.2cm]      (c2)       [below left = \tridist and 0.4\tridist of c1] {};
      \node[roundnode, minimum size=.2cm]      (c3)       [below right = \tridist and 0.4\tridist of c1] {};

      \draw (c1) -- (c2) node[left, pos=0.5]{+};
      \draw (c1) -- (c3) node[right, pos=0.5]{+};
      \draw (c2) -- (c3) node[below, pos=0.5]{-};

      \node[roundnode, minimum size=.2cm]      (d1)       [right=1.75\tridist of c1]                      {};
      \node[roundnode, minimum size=.2cm]      (d2)       [below left = \tridist and 0.4\tridist of d1] {};
      \node[roundnode, minimum size=.2cm]      (d3)       [below right = \tridist and 0.4\tridist of d1] {};

      \draw (d1) -- (d2) node[left, pos=0.5]{-};
      \draw (d1) -- (d3) node[right, pos=0.5]{-};
      \draw (d2) -- (d3) node[below, pos=0.5]{-};

    \end{tikzpicture}    
    \caption{The four possible signed triangles. The two on the left are balanced, while the two on the right are not.}
    \label{fig:triangles}
  \end{center}
\end{figure}

Signed graphs can be used to overcome this limitation. In signed graphs, each edge is labeled with either a positive or negative sign. If a graph represents social interactions, signs can be employed to determine whether these interactions are friendly or not. Thus, signed graphs constitute a good representation for detecting polarized groups in online debates.
%
Signed graphs were first introduced by Harary to study the concept of {\em balance}~\cite{harary1953notion}. A signed graph is said to be balanced if its vertices can be partitioned into two sets in perfect agreement with the edge signs; that is, every edge within each set is positive and every edge between the two sets is negative. Equivalently, a signed graph is balanced when the product of the signs of every cycle is positive. This is analogous to the commonplace notion ``the friend of a friend is a friend,'' ``the enemy of a friend is an enemy,'' etc., as illustrated in Fig. \ref{fig:triangles}.
A substantial body of work has been devoted to studying the spectral properties of signed graphs, which have strong connections to the concept of balance. In particular, the spectrum of the Laplacian matrix of a signed graph reveals whether it is balanced \cite{kunegis2010spectral}. Graphs found in real applications are often not balanced, and therefore the question of finding a balanced subgraph arises naturally.
The problem of finding a \emph{maximum balanced subgraph} be formulated in terms of vertex cardinality (\mbs) or edge cardinality (\mbse). 
Both formulations lead to \NPhard problems, 
thus, the development of efficient heuristics to approximately solve this problem is well motivated.

In this paper we present an algorithm to find large balanced subgraphs in signed networks. The algorithm works in two stages. First, we rely on spectral theory, as well as on a novel bound for perturbations of the Laplacian, to develop a greedy method to remove vertices and uncover a balanced subgraph. Then, any removed vertices that do not violate the balance of the located structure are restored. We derive analytical properties that allow us to efficiently implement the algorithm. Finally, we devise a random sampling strategy to significantly enhance the scalability of the method.

In a variety of experiments on real-world and synthetic data, we show that the proposed algorithm finds larger balanced subgraphs than alternative heuristics from the literature, both in terms of vertex and edge count. Furthermore, the proposed algorithm has the desirable properties that 
(i) it runs faster than any other competing method, 
(ii) it can be tuned to trade off running time and solution quality, and 
(iii) produces a vertex-removal sequence, which can be used to trade off balance and graph size. We validate the scalability of the method by testing it on graphs of up to 34 million edges.

\noindent
Our contributions are summarized as follows:
\squishlist
\item We propose an algorithm for finding large balanced subgraphs in signed graphs, based on spectral theory and its connections to balance.
\item We give an upper bound for the smallest Laplacian eigenvalue after removing a set of vertices, which allows us to efficiently trim the graph to find a balanced subgraph.
\item We experimentally show that our algorithm finds subgraphs much larger than those found by state-of-the-art methods.
\item We devise a random sampling strategy to significantly enhance the scalability of the method, and show empirically that the quality of the output is not affected.
\squishend
The rest of this paper is structured as follows. We discuss related work in Section \ref{section:related}. In Section \ref{sec:preliminaries} we introduce our notation and relevant notions. In Section \ref{sec:algorithm} we describe our algorithm in detail and discuss relevant considerations, and in Section \ref{sec:experiments} we show our experimental results. Finally, Section \ref{sec:conclusion} is a short conclusion.

\section{Related work}
\label{section:related}

\spara{Signed graphs and balance theory.}
Signed graphs were first studied by Harary,
who was interested in the notion of {\em balance}~\cite{harary1953notion}.
In 1956, Cartwright and Harary generalized Heider's psychological theory
of balance in triangles of sentiments to the theory of balance in
signed graphs~\cite{cartwright1956structural}.
Early work on signed graphs focused mainly on properties related to balance theory.
For example, Harary and Kabell develop a simple linear-time algorithm
to test whether a given signed graph satisfies the {\em balance property}~\cite{harary1980simple};
while Akiyama et al.~\cite{akiyama1981balancing} study properties
of the minimum number of sign changes required so that a signed graph
satisfies the balance property.

A more recent line of work
develops spectral properties of signed graphs,
still related to the balance theory.
Hou et al.~\cite{hou2003laplacian} prove that a connected signed graph is balanced
if and only if the smallest eigenvalue of the Laplacian is~0.
Hou~\cite{hou2005bounds} also investigates the relationship between the smallest eigenvalue of the Laplacian
and the unbalancedness of a signed graph.

\spara{Maximum balanced subgraphs.}
The problem studied in this paper is to find a maximum balanced subgraph (\mbs) 
in a given signed graph. 
Poljak and Turz{\'\i}k~\cite{poljak1986polynomial} show
that every connected signed graph with $n$ vertices and $m$ edges has a balanced
subgraph with at least $\frac{m}{2} + \frac{n-1}{4}$ edges, and this bound is tight. 
They give an algorithm to find such a subgraph that requires at least $\bigO(n^3)$ computations.
Notice that this algorithm gives a 2-approximation for the \mbse\ problem, 
but it is not practical for large graphs.
The \mbse\ problem can be formulated as a {\sc Signed\-Max\-Cut} problem, 
which is a generalization of the standard {\sc Max\-Cut} problem, and thus, \NP-hard. 
To obtain an exact solution, the problem has been studied in the context of 
{\em fixed-parameter tractability} (\fpt).
H{\"u}ffner et al.~\cite{huffner2007optimal}
propose an \fpt algorithm for deciding whether the maximum balanced subgraph has size at least $m-k$, 
where $k$ is the parameter.
Motivated by the lower bound of Poljak and Turz{\'\i}k, Crowston et al.~\cite{crowston2013maximum}
give an \fpt algorithm for deciding whether the maximum balanced subgraph has at least 
$\frac{m}{2} + \frac{n-1}{4} + \frac{k}{4}$ edges, where $k$ is the parameter.
These algorithms are not practical, as their running time is exponetial in $k$
and the degree of the polynomial in $n$ is large.

The \mbse\ problem has also been considered in application-driven studies, 
and different heuristics have been proposed.
DasGupta et al.~\cite{dasgupta2007algorithmic} consider an edge-deletion formulation of the \mbse\ problem
in the context of biological networks. 
Motivated by the {\sc Max\-Cut} connection, 
the authors develop an algorithm based on semidefinite programming relaxation (\sdp);
the approach, however, is not scalable and tested only on very small networks.
Figueiredo and Frota~\cite{figueiredo2014maximum} ask to find a balanced subgraph that
maximizes the number of vertices. 
They propose and experiment with a branch-and-cut exact approach, 
a heuristic based on minimum-spanning tree computation, and a heuristic combining a greedy algorithm and local-search.
We experimentally compare the proposed method with these heuristics in our empirical evaluation. 

\spara{Community detection in signed graphs.}
Different approaches have been proposed for community detection in signed graphs, 
some of which try to incorporate balance theory. 
Anchuri et al.~\cite{anchuri2012communities}
propose a spectral approach to partition a signed graph 
into a number of non-overlapping balanced communities.
Yang et al.~\cite{yang2007community} propose a random-walk-based approach
for partitioning a signed graph into communities, 
where in addition to positive edges within clusters and negative edges across clusters, 
they also want to maximize cluster densities.
Doreian and Mrvar~\cite{doreian1996partitioning} propose an algorithm 
for partitioning a signed \emph{directed} graph so as to minimize a measure of \emph{imbalance}.
The approach is evaluated only on very small networks.
Signed directed graphs are also considered by Lo et al.~\cite{lo2011mining}, 
who search for strongly-connected positive subgraphs that are negative bi-cliques. 
Chu et al.~\cite{chu2016finding} propose a constrained-programming objective to find
$k$~\emph{warring factions}, as well as an efficient algorithm to find local optima. 
Bonchi et al.\ formulate the problem of finding subgraphs in signed networks that are dense 
but allow for imperfect balance~\cite{bonchi2019discovering}. 
Several other methods have been proposed for identifying communities in signed graphs, 
some of which incorporating notions related to balance. 
A detailed survey on those methods is provided by Tang et al.~\cite{tang2016survey}.
The main difference of our work in comparison with all these approaches
is that they are mainly interested in communities or graph partitioning, 
which are different than the \mbs problem.

\section{Preliminaries}
\label{sec:preliminaries}
Before describing the proposed algorithm, we introduce our notation and review some basic results from the literature.

We consider an undirected simple signed graph $\graph=(V,\posE,\negE)$ where $V=\{1, \dots, n\}$ is the set of vertices and $\posE$ (respectively, $\negE$) is the set of positive (respectively, negative) edges. We sometimes simplify this notation and write $\graph=(V,E)$, where $E=\posE \cup \negE$.
Throughout this paper we denote vectors with boldface letters ($\vec{v}$) and matrices with uppercase letters ($A$). We use $\vec{v}_i$ to denote the $i$-th entry of a vector $\vec{v}$, and $A_{ij}$ to denote the element in the $i$-th row and $j$-th column of matrix $A$.
Given a signed graph, we define its adjacency matrix $A$ as follows: $A_{ij}=1$ if $\{i,j\}\in \posE$, $A_{ij}=-1$ if $\{i,j\}\in \negE$ and 0 otherwise. Further, we define the diagonal degree matrix $D$ as $D_{ii}=d(i)$, where $d(i)$ is the degree of vertex $i$, i.e., the number of edges (either positive or negative) adjacent to $i$. The \emph{signed Laplacian} of $\graph$ is defined as $\lap(\graph)=D-A$. We also refer to this matrix simply as \emph{Laplacian}, and will denote $L=L(G)$ when there is no ambiguity. Given a set of vertices $S$ such that $S\subseteq V$, $G\setminus S$ denotes the graph that results from removing from $G$ the vertices in $S$, as well as all adjacent edges.

We now define the concept of \textit{balance in signed networks}, which is central to our paper.

\begin{definition}[Balanced graph]
Given a connected signed graph $\graph=(V,\posE,\negE)$, $\graph$ is balanced if there exists a partition $V=V_1\cup V_2$, $V_1\cap V_2=\emptyset$ such that every edge with both endpoints in $V_1$ is positive, every edge with both endpoints in $V_2$ is positive, and every edge with one endpoint in $V_1$ and the other in $V_2$ is negative.
\end{definition}
In other words, a graph is balanced if we can divide it into two sets in a way that every edge sign agrees with the partition. For instance, if the vertices in graph $\graph$ represent users in a social network and the edges interactions between them (friendly or hostile, depending on the sign), a dense, balanced graph suggests that there are two polarized communities.

It is easy to decide whether a given signed graph is balanced.
In addition to a simple combinatorial algorithm, 
there is also an interesting characterization of balanced graphs based on the spectrum of the signed Laplacian. 
This is shown by the following result, which is key in the derivation of our algorithm.
\begin{theorem}[\cite{hou2003laplacian}]
  \label{the:balanced}
Consider a connected signed graph $\graph=(V,\posE,\negE)$, with signed Laplacian $\lap$. Let $\lambda_1(\lap)\leq \dots, \leq\lambda_n(\lap)$ be the eigenvalues of $\lap$. Then $\graph$ is balanced if and only if $\lambda_1(\lap)=0$.
\end{theorem}

The smallest eigenvalue of the Laplacian reveals not only whether a graph is balanced, but also how far it is from being balanced. This is established by the following result of Li and Li~\cite{li2016note}.
\begin{theorem}[\cite{li2016note}]
  \label{the:how_balanced}
  Consider a connected signed graph $\graph=(V,\posE,\negE)$, with signed Laplacian $\lap$. Let $\lambda_1(\lap)\leq \dots, \leq\lambda_n(\lap)$ be the eigenvalues of $\lap$. Then
  \[
  \lambda_1(\lap)\leq \min_{G'}\{\lambda_n(L(G')):V_{G'}\subseteq V, G\setminus V_{G'} \text{ is balanced}\}.
  \]
\end{theorem}
Here, $V_{G'}$ denotes the set of vertices of graph $G'$.
Intuitively, Theorem~\ref{the:how_balanced} says that if we can make $G$ balanced with just a small modification, then $\lambda_1(\lap)$ is small. Note that  $\lambda_n(L(G'))\leq 2\Delta(G')$, where $\Delta(G')$ denotes the maximum degree of the vertices of $G'$.

Graphs found in practical applications are usually not balanced. The question that arises naturally is thus whether we can find the maximum balanced subgraph (\mbs) of a given signed graph efficiently. We formalize this objective next.

\begin{problem}[\mbs]
\label{prob:mbs}
  Given a signed graph $G=(V,E)$, find the graph $G'$ induced by $V'\subseteq V$ such that $G'$ is balanced and the cardinality of\ $V'$\ is maximized.
\end{problem}
A solution to Problem \ref{prob:mbs} would reveal the frustration number, that is, the minimum number of vertices to remove to make the graph balanced, and is thus \NPhard \cite{zaslavsky2012mathematical}. In this paper we approach this problem based on Theorems \ref{the:balanced} and \ref{the:how_balanced}. In particular, we address the following question: what vertices should we remove from $G$ so that the minimum eigenvalue of the resulting graph is as small as possible? This question inspires the algorithm to find balanced subgraphs described in the next section.

\section{Algorithm}
\label{sec:algorithm}
Our algorithm works in two stages. First, it greedily removes vertices from the graph to improve balance as much as possible, until it obtains a balanced subgraph. Then, it does a single pass over the removed vertices and restores the ones that do not violate balance. In this section we describe these two stages in detail, as well as several optimizations and a procedure to enable the processing of huge graphs. Throughout this section, we assume the input graph to be connected.

The procedure, which we dub \ouralgo (Trimming Iteratively to Maximize Balance), is summarized in Algorithm \ref{alg:timbal}. This section explains each of its steps in detail.

\begin{algorithm}[t]
  \caption{\ouralgo}
  Input: signed graph $\graph$
  \begin{algorithmic}[1]
    \STATE $R\gets \emptyset$
    \STATE Optionally: $R\gets$ \randprepro$(\graph)$, $\graph \gets \graph \setminus R$    
    \WHILE {$G$ is not balanced}
    \STATE Compute $L$, the signed Laplacian of $G$.
    \STATE Compute bound vector $\rankvec$.
    \STATE Compute the set of vertices to remove $S$, based on $\rankvec$.
    \STATE $G \gets G\setminus S$; $R\gets R\cup S$.
    \STATE $G \gets $ largest connected component in $G$.
    \ENDWHILE
    \FOR {$v \in R$}
    \IF {$G\cup \{v\}$ is balanced}
    \STATE $G \gets G\cup \{v\}$
    \ENDIF
    \ENDFOR 
    \STATE Output $G$
  \end{algorithmic}
  \label{alg:timbal}
\end{algorithm}

\subsection{First stage: removing vertices}
\label{sec:first_stage}
In the first stage of the algorithm, we iteratively remove vertices from the graph. The challenge is to determine which vertices to remove at each step. Our criterion for selecting vertices to remove is based on Theorem~\ref{the:balanced}, and more precisely on Theorem~\ref{the:how_balanced}. In particular, the smallest eigenvalue of the signed Laplacian measures how far the graph is from being balanced.

Given a graph $\graph$ with signed Laplacian $\lap$, define $\lapi$ to be the signed Laplacian of $\graph\setminus \{i\}$, that is, of the graph resulting from removing vertex $i$. We want to find the vertex that minimizes the smallest eigenvalue of the resulting Laplacian, that is, we want to find the vertex $j$ such that
\begin{align}
  \label{eq:eig_criterion}
  j=\argmin_i\lambda_1(\lapi).
\end{align}

Naturally, computing $\lambda_1(\lapi)$ for every vertex $i$ using the spectral decomposition of the modified Laplacian $\lapi$ is costly. To overcome this challenge, we present the following result, which gives an upper bound on the smallest eigenvalue of the perturbed Laplacian.

\begin{lemma}
  \label{lem:upper_bound}
  Given a graph $\graph$ with signed Laplacian $\lap$, let $\lambda_1(\lap)$ be the smallest eigenvalue of $\lap$ and $\vec{v}$ an eigenvector of $\lap$ satisfying $L\vec{v}=\lambda_1(\lap)\vec{v}$. Then
  \begin{align}
    \label{eq:upper_bound}
  \lambda_1(\lapi) \leq \frac{\lambda_1(\lap)(1-2\vec{v}_i^2)-\sum_{j\in N(i)}\vec{v}_j^2 + \vec{v}_i^2d(i)}{1-\vec{v}_i^2},
  \end{align}
\end{lemma}
where $N(i)$ denotes the set of neighbours of $i$ in $G$.

\begin{proof}
  We can obtain the matrix $\lapi$ by applying the following operations to $\lap$: (1) for every $j$ in $N(i)$, subtract 1 from $\lap_{jj}$; (2) remove row and column $i$. Thus, if we define the vector $\vec{\hat v}$ to be equal to $\vec{v}$ after removing the $i$-th entry, it is
  \[
    \vec{\hat v}^T\lapi\vec{\hat v} = \vec{v}^T\lap \vec{v} - d(i)\vec{v}_i^2 -\sum_{j\in N(i)}\vec{v}_j^2 - 2\vec{v}_i\sum_{j\in N(i)}\vec{v}_jL_{ij}.
  \]
  Now, observe that $\vec{v}^T\lap \vec{v}=\lambda_1(\lap)$ and $\sum_{j\in N(i)}\vec{v}_jL_{ij}=\lambda_1(\lap)\vec{v}_i-d(i)\vec{v}_i$. Therefore, 
  \[
    \vec{\hat v}^T\lapi\vec{\hat v} =  \lambda_1(\lap) - \left ( \sum_{j\in N(i)}\vec{v}_j^2 + \vec{v}_i^2(2\lambda_1(\lap)-d(i)) \right).
  \]
  Since $\lambda_1(\lapi)=\min_{\vec{x}}\frac{\vec{x}^T\lapi \vec{x}}{\vec{x}^T\vec{x}}$, we have
  \begin{align*}
    \lambda_1(\lapi) &\leq \frac{\vec{\hat v}^T\lapi\vec{\hat v}}{\vec{\hat v}^T\vec{\hat v}}
    \\&=  \frac{\lambda_1(\lap) - \left ( \sum_{j\in N(i)}\vec{v}_j^2 + \vec{v}_i^2(2\lambda_1(\lap)-d(i)) \right)}{1-\vec{v}_i^2}.
  \end{align*}
  Elementary computations yield the desired result.  
\end{proof}
We define the vector $\rankvec$ whose entries are the values of the right-hand side of Inequality~(\ref{eq:upper_bound}), for each $i$, i.e.,
\begin{align}
  \label{eq:rankvec}
  \rankvec_i=\frac{\lambda_1(\lap)(1-2\vec{v}_i^2)-\sum_{j\in N(i)}\vec{v}_j^2 + \vec{v}_i^2d(i)}{1-\vec{v}_i^2}.
\end{align}
In order to choose which vertex to remove from the graph, we can therefore take the one minimizing $\rankvec_i$. The first stage of our algorithm removes vertices according to this criterion --- as shown in Section~\ref{sec:several}, we can remove several vertices at once --- until a balanced subgraph is found.

\subsection{Second stage: restoring vertices}
Once we have found a balanced induced subgraph $G'=(V',E')$, it is straightforward to obtain the corresponding partition $V_1\cup V_2=V'$ that agrees perfectly with the edge signs, e.g., by performing a breadth-first search. At this point, we can inspect the vertices removed in the first stage to see if some of them agree with the obtained partition and they can thus be reinserted into the graph. That is, if our input graph is $\graph=(V,E)$, we consider the set of removed vertices $R=V\setminus V'$. For every $v\in R$, if adding $v$ to either $V_1$ or $V_2$ results in a balanced graph, we add it back to $G'$ --- restoring as well its edges with endpoints in $G'$ --- and proceed. We inspect this set of vertices in the order they were removed from the graph in the first stage of the algorithm.

In the remainder of this section we discuss the different considerations that must be taken into account to implement our algorithm.

\subsection{Computing the bound efficiently}
\label{sec:efficiency}
A key advantage of the bound from Lemma~\ref{lem:upper_bound} is that it can be computed efficiently. Given a graph $\graph$, we define $\bar\lap$ to be the matrix whose entries are the absolute values of the entries of $\lap(\graph)$. 

Define the matrix $W=\bar\lap+2\lambda_1(\lap)I$ and $\vec{w}=\vec{v}\circ \vec{v}$, where $\circ $ denotes the element-wise product of two vectors. Then
\begin{align}
  \label{eq:boundfast}
\vec{r}_i=\frac{(\lambda_1(\lap)\mathbf 1- W\vec{w})_i}{1-\vec{v}_i^2}.
\end{align}
That is, the computation of the vector $\vec{r}$ reduces essentially to a matrix-vector multiplication operation.


\subsection{Removing several vertices at once}
\label{sec:several}

Every time we remove a vertex, we need to compute the smallest eigenvalue and corresponding eigenvector of the updated Laplacian. Even though this can be done efficiently (see Section~\ref{sec:update_ed}), when dealing with large graphs the overall computation time may be too high. Therefore, it might be desirable to remove several vertices at once, before updating the eigenpair, in order to find a balanced subgraph more quickly.

The most straightforward way to accomplish this batch operation is to simply consider the $k$ smallest entries of $\rankvec$ and remove the corresponding $k$ vertices. However, we argue that this might have undesirable consequences. Consider the graph on the left of Figure~\ref{fig:example_many}. Removing either vertex 1 or 2 will make the graph balanced. Thus, both $\rankvec_1$ and $\rankvec_2$ --- where $\rankvec$ is the ranking vector defined in Equation~(\ref{eq:rankvec}) --- are bound to be equally small. Removing these two vertices at the same time will result in a subgraph of size 2, but we could have obtained a balanced subgraph of size 3 by removing only one vertex.

To partially alleviate this shortcoming, we propose considering independent --- i.e., non-neighbouring --- vertices for simultaneous removal only. Formally, consider we want to remove $k$ vertices, and assume we have so far chosen $q<k$ of these, to form the set $R$. Then, the next chosen vertex is defined as
\[
\argmin_{i\notin \bigcup_{j\in R}N(j)}\rankvec_i.
\]
Choosing independent vertices has the additional advantage that the bound given in Lemma~\ref{lem:upper_bound} is additive in the following sense.

\begin{lemma}
  \label{lem:upper_bound_k}
  Given a graph $\graph$ with signed Laplacian $\lap$, let $\lambda_1(\lap)$ be the smallest eigenvalue of $\lap$ and $\vec{v}$ an eigenvector of $\lap$ satisfying $L\vec{v}=\lambda_1(\lap)\vec{v}$. Let $R$ be a set of independent vertices in $G$ and let $\lapR$ be the signed Laplacian of $G\setminus R$. Then
  \begin{align}
    \label{eq:upper_bound_k}
  \lambda_1(\lapR) \leq \frac{\sum_{i\in R}\lambda_1(\lap)(1-2\vec{v}_i^2)-\sum_{j\in N(i)}\vec{v}_j^2 + \vec{v}_i^2d(i)}{1-\sum_{i\in R}\vec{v}_i^2}.
  \end{align}
\end{lemma}
The lemma is easily verified by similar techniques as employed in the proof of Lemma~\ref{lem:upper_bound}.

Furthermore, the value of this upper bound can be tracked as we add vertices to the set $R$ to decide how many of them to remove. Note that as the set $R$ grows, the upper bound becomes less reliable (as the perturbation of $L$ is more significant). At some point, the magnitude of the denominator will become too small and the bound will decrease very slowly, or even increase. This can be used as a criterion to choose a cut-off point for the removal.

Note that this does not completely solve the problem of adequately choosing a vertex set for simultaneous removal. Consider, for instance, a cycle graph with arbitrary signs. Removing any vertex results in a balanced subgraph, since every tree is balanced \cite{zaslavsky1982signed}. However, simply discarding neighbouring vertices is not sufficient to limit the number of removed vertices to one in this case. In general, determining this set to optimality might constitute a hard problem in and of itself, and is therefore left for future work. Nevertheless, in our experiments we show that discarding neighbours provides good results in practice.

\begin{center}
  \begin{figure}[t]
    \centering
    \begin{tikzpicture}[
        roundnode/.style={circle, draw=gray!60, fill=gray!20, very thick, minimum size=7mm},
      ]
      \node[roundnode]      (ur)                              {3};
      \node[roundnode]        (lr)       [below=of ur] {4};
      \node[roundnode]      (ul)       [left=of ur] {1};
      \node[roundnode]        (ll)       [below=of ul] {2};
      
      \draw[-] (ur.south) -- (lr.north);
      \draw[dotted, thick] (ul) -- (ur);
      \draw[dotted, thick] (ll) -- (lr);
      \draw[dotted, thick] (ul) -- (lr);
      \draw[dotted, thick] (ll) -- (ur);
      \draw[dotted, thick] (ll) -- (ul);
      
      \draw[|->] (1,-1) -- (2,-1);
      
      \node[roundnode]      (ul2)       [right=2cm of ur]{1};
      \node[roundnode]      (ur2)         [right=of ul2]           {3};
      \node[roundnode]        (lr2)       [below=of ur2] {4};

      \draw[dotted, thick] (ul2) -- (ur2);
      \draw[dotted, thick] (ul2) -- (lr2);
      \draw[-] (ur2) -- (lr2);
    \end{tikzpicture}
    
    \caption{Illustration of why we must only consider independent vertices for simultaneous removal. Solid edges are positive, while dashed ones are negative. Removing either 1 and 2 will make the graph balanced, so both will get an equal value in the ranking. However, we only need to remove one.}
    \label{fig:example_many}
  \end{figure}
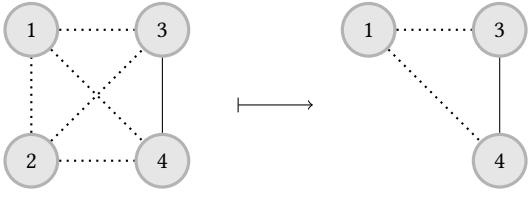
\end{center}

\subsection{Updating the eigenpair}
\label{sec:update_ed}
A remaining concern is the computation of the eigenvalue $\lambda_1(\lapi)$ and the corresponding eigenvector each time we remove a set of vertices. We propose two alternatives for this purpose.

\spara{Locally-optimal block preconditioned conjugate gradient me\-thod:} Since the smallest eigenvalue of the Laplacian is in the ``flat'' part of the spectrum, that is, where consecutive eigenvalues are close to each other, the methods usually employed to compute eigenvalue decompositions can be slow to converge when dealing with large matrices. To speed up the process, we use the method described by Knyazev \cite{knyazev2001toward} to estimate the desired eigenpair. Since our goal is to rank the vertices according to the corresponding upper bound from Lemma~\ref{lem:upper_bound}, an approximation of $\lambda_1(\lapi)$ and the corresponding eigenvector is enough. Our experimental results support this claim.

\spara{Low-rank updates of the eigenvalue decomposition:} An alternative to the use of the conjugate gradient method is to rely on known results concerning low-rank updates of the spectral decomposition. In particular, consider a matrix $L$ with eigenvalue decomposition $L=Q\Lambda Q^T$. Now consider a matrix $\tilde L=L+W$, where $W$ is a rank-$\rho$ matrix satisfying $W=VV^T$, i.e., positive semidefinite. It is well known \cite{arbenz1988spectral} that an eigenvalue $\lambda$ of $\tilde L$ --- not in the spectrum of $L$ --- makes the following expression equal to zero:
\begin{align}
\det(I-V^T(\lambda-L)^{-1}V),
\end{align}
where $\det$ denotes the determinant of a matrix.
Elementary operations yield the following, equivalent expression:
\begin{align}
  \label{eq:ed_update}
\det(I-U^T(\lambda-\Lambda)^{-1}U),
\end{align}
where $U=Q^TV$. Since $(\lambda-\Lambda)^{-1}$ is diagonal, its inversion is cheap. Moreover, the number of eigenvalues of $\tilde L$ below any real number can be inferred exactly.
That is, if one can afford to compute the eigenvalue decomposition of $L$ and the rank $\rho$ of the perturbation $W$ is small, the eigenvalue decomposition of $\tilde L$ can be computed efficiently by means of a bisection algorithm on Expression~(\ref{eq:ed_update}).

It only remains to show how to compute the matrix $V$, so that we can construct $U$. We now show that if the vertices to remove are chosen carefully, the entries of $V$ depend only on the degree of the removed vertices, and thus can be permanently stored and queried on execution, instead of being computed at each iteration.

First, consider a graph $\graph$ with signed Laplacian $\lap$. We remove a single vertex and obtain a new graph with Laplacian $\lapi=\lap-W$. It is easily verified that $W$ --- if we exclude zero rows and columns --- is the Laplacian of a star graph, that is, a connected $k$-tree with $k-1$ leaves. Now, consider we remove a set $R$ of vertices, to obtain $\lapR=\lap-W$, satisfying the following condition:
\begin{align}
\label{eq:indep_condition}
  \text{for all } i\neq j \in R, N(i)\cap N(j)=\emptyset.
\end{align}
Then the matrix $W$ can be permuted so that it is block-diagonal, with each block corresponding to the Laplacian of a star graph. The following result establishes that its eigenvalues are easily inferred from the size of each block.
\begin{lemma}
  Let $\graph$ be a signed star graph comprised of $k$ vertices, with signed Laplacian $\lap$. Then the eigenvalues of $\lap$ are
  \begin{itemize}
  \item 0 with multiplicity 1,
  \item $k$ with multiplicity 1 and
  \item 1 with multiplicity $k-2$.
  \end{itemize}
\end{lemma}
\begin{proof}
  First, since every tree is balanced \cite{zaslavsky1982signed}, from Theorem~\ref{the:balanced} we know that 0 is an eigenvalue of $\lap$.

  To infer the rest of the eigenvalues, consider the structure of $\lap$: it is $\lap_{11}=k-1$, and $\lap_{ii}=1, \lap_{1i}=\lap_{i1}=\pm 1$ for all $i>1$. Thus, the vector $\vec{v}=(x, \gamma_2, \dots, \gamma_k)$ where $x$ is an arbitrary real number and  $\gamma_i=x/(k-1)\times \sign{\lap_{i1}}$ is an eigenvector of eigenvalue $k$.
  That 1 is an eigenvalue with multiplicity $k-2$ is easily verified from the fact that the rank of $\lap-I$ is 2.
\end{proof}
Given the structure of $\lap$, once the eigenvalues are known, the eigenvectors can be efficiently computed. Thus, the eigenvalue decomposition of $W$ need not be computed explicitly.
The eigenvectors only need to be computed once for each value of $k$, and can then be reused whenever a vertex of degree $k$ is removed.

To summarize, if we know the spectral decomposition of $L$ and choose to remove vertices satisfying Condition (\ref{eq:indep_condition}), we can easily make use of Expression~(\ref{eq:ed_update}) to update the eigenpair.

\subsection{Handling various connected components}
During the execution of the first stage --- see Section~\ref{sec:first_stage} --- after removing the chosen vertices the graph might become disconnected. In this case, we can consider various alternatives. If among the resulting connected components only one is large enough, we can discard the rest. If various connected components are large enough to warrant further analysis, we can simply apply the algorithm recursively on each of them.

Nevertheless, in our experiments we observed that the connected components resulting from the removal, except from the largest one, are small, comprised of a handful of vertices in almost every case. This is consistent with the principle behind the algorithm --- i.e., the reduction of the smallest eigenvalue of the Laplacian. Note, for instance, that graphs comprised of two vertices are always balanced. Therefore, if removing a vertex results in a two-vertex structure becoming disconnected, the obtained graph will have a smallest eigenvalue equal to zero, and this vertex will thus be a good choice according to the criterion defined by (\ref{eq:eig_criterion}). Notice that this is not undesirable behaviour: if a small subgraph becomes disconnected by removing one (or a few) vertices, then it is not densely connected to the rest of the graph and would therefore not contribute significantly to the density of the balanced subgraph found by the algorithm.

\subsection{Scaling to big graphs}
\label{sec:scaling}
The proposed algorithm provides an efficient method to detect vertices to remove from a graph in order to improve its balance. The method, however, requires the estimation of an eigenvalue-eigenvector pair from the ``flat'' part of the spectrum, where algorithms for this purpose take longer to converge. This operation is linear in the number of entries of the adjacency matrix, that is, potentially quadratic in the number of vertices. Thus, the analysis of a graph comprised of millions of nodes can remain impractical.

To alleviate this shortcoming we propose a randomized preprocessing algorithm. Given the quadratic complexity of our algorithm's iterations, processing a large number of subgraphs independently and then combining the results can result in significant time savings. To justify the approach, we rely on the following, easily verified, statement:
\begin{proposition}
  Let $G$ be a signed graph. $G$ is balanced if and only if all of its subgraphs are balanced.
\end{proposition}
The statement guarantees, on one hand, that by balancing random subgraphs we will never remove vertices from a balanced substructure of $G$, and therefore the process is safe in this regard. On the other hand, it implies that we cannot obtain a balanced subgraph of $G$ before all of its subgraphs are balanced. Thus, any unbalanced subgraph we sample needs to become balanced in $G$. Based on these facts, we propose the preprocessing algorithm summarized as Algorithm \ref{alg:randprepro}. The algorithm randomly samples a large number of connected subgraphs and then runs Algorithm \ref{alg:timbal} on them. The set of vertices that have been removed from at least one of these subgraphs is then removed from the main graph, which is then processed normally.

To sample subgraphs, we take the following approach. We sample a vertex uniformly at random, and then perform a randomized breadth-first search (RBFS). RBFS works as follows: we first take all neighbours of the sampled vertex. Then, at each step of the search, we only take a random fraction of the corresponding vertex's neighbours. The size of the fraction and the depth of the search are set by the user to achieve subgraphs of a certain size. This sampling strategy can produce dense subgraphs very efficiently.

\begin{algorithm}[t]
  \caption{\randprepro}
  Input: signed graph $\graph$
  \begin{algorithmic}[1]
    \STATE Sample $s$ subgraphs $G_1, \dots, G_s$ of size $k$ from $G$.
    \FOR {$i=1\dots, s$}
    \STATE Run first stage of \ouralgo on $G_i$, resulting in set of vertices to remove $R_i$.
    \ENDFOR
    \STATE Output $\bigcup_i R_i$.
  \end{algorithmic}
  \label{alg:randprepro}
\end{algorithm}

\subsection{Complexity analysis}
The running time of our algorithm is dominated by the following operations: the computation of the smallest eigenvalue and corresponding eigenvector of the Laplacian, done essentially in $\bigO(|E|)$ operations; the computation of the bounds in Eq. (\ref{eq:upper_bound}), which can be done in time $\bigO(|V|^2)$, as shown in Eq. (\ref{eq:boundfast}); the test of balance and location of connected components, a breadth-first search done in $\bigO(|V|+|E|)$ time operations; the vertex restoration phase can be implemented using two inner product operations per vertex to check compliance with the balanced subgraph, which results $\bigO(|V|\Delta)$ operations for the whole stage, where $\Delta$ is the maximum degree over all graph vertices. All these steps can be implemented to exploit the efficiency of sparse matrix operations.

The optional \randprepro procedure computes the first stage of the algorithm for each sampled subgraph, whose size can be controlled by limiting the number of sampled neighbours in the RBFS step. Our experimental results shows that this approach remains effective even when the size of the sampled subgraphs is small, and their number moderate.

\section{Experiments}
\label{sec:experiments}

This section presents the evaluation of the proposed algorithm. Our main purpose is to determine the effectiveness of our method in finding large balanced subraphs. In particular, we evaluate the following aspects:
\squishlist
\item We measure the size of the subgraphs found by our method, in both vertex and edge cardinality.
\item We assess whether the vertex removal sequence produced by our method can be exploited to trade off solution quality and size, where by quality we mean that we allow our method to return not perfectly balanced subgraphs.
\item We measure the running time of our implementation.
\item We visualize some of the obtained results.
\squishend

For our experimental evaluation, we use a variety of real-world data, which we briefly describe below.

\spara{Datasets.}
We select publicly-available real-world signed networks, whose main characteristics are summarized in Table~\ref{tab:datasets}.
\textsf{HighlandTribes}\footnote{\href{http://konect.cc}{konect.cc}\label{foot:k}} represents the alliance structure of the Gahuku--Gama tribes of New Guinea.
\textsf{Cloister}$^{\ref{foot:k}}$ contains the esteem/disesteem relations of monks living in a cloister in New England (USA).
\textsf{Congress}$^{\ref{foot:k}}$ reports (un/)favorable mentions of politicians speaking in the US Congress.
\textsf{Bitcoin}\footnote{\href{http://snap.stanford.edu}{snap.stanford.edu}\label{foot:s}} and \textsf{Epinions}$^{\ref{foot:s}}$ are who-trusts-whom networks of the users of Bitcoin OTC and Epinions, respectively.
\textsf{WikiElections}$^{\ref{foot:k}}$ includes the votes about admin elections of the users of the English Wikipedia.
\textsf{Referendum}~\cite{lai2018stance} records Twitter data about the 2016 Italian Referendum: an interaction is negative if two users are classified with different stances, and is positive otherwise.
\textsf{Slashdot}$^{\ref{foot:s}}$ contains friend/foe links between the users of Slashdot.
The edges of \textsf{WikiConflict}$^{\ref{foot:s}}$ represent positive and negative edit conflicts between the users of the English Wikipedia.
\textsf{WikiPolitics}$^{\ref{foot:k}}$ represents interpreted interactions between the users of the English Wikipedia that have edited pages about politics.

\begin{table}[t]
\centering
\caption{\label{tab:datasets}Signed networks used in the experiments:
number of vertices and edges;
ratio of negative edges ($\rho_- = \frac{|E_-|}{|E_+ \cup E_-|}$);
 and ratio of non-zero elements of $A$ ($\delta = \frac{2|E_+ \cup E_-|}{|V|(|V|-1)}$).}
\vspace{-3mm}
\centerline{
\small
\begin{tabular}{lrrrrr}
\toprule
Real-world datasets & $|V|$    &    $|E_+ \cup E_-|$    &    $\rho_-$    &   $\delta$\\
\midrule
\tribes &  $16$\,\,\, &  $58$\,\,\,\, & $0.50$    &    $0.48$\,\,\,\,\,\,\,\,\,\,\\
\cloister       &  $18$\,\,\, & $125$\,\,\,\, & $0.55$    &    $0.81$\,\,\,\,\,\,\,\,\,\,\\
\congress       & $219$\,\,\, & $521$\,\,\,\, & $0.20$       &    $0.02$\,\,\,\,\,\,\,\,\,\,\\
\bitcoin        &   $5$\,k    &  $21$\,k\,    & $0.15$       &    $1.2e\!-\!03$\\
\election  &   $7$\,k    & $100$\,k\,    & $0.22$       &    $3.9e\!-\!03$\\
\referendum     &  $10$\,k    & $251$\,k\,    & $0.05$       &    $4.2e\!-\!03$\\
\slashdot       &  $82$\,k    & $500$\,k\,    & $0.23$       &    $1.4e\!-\!04$\\
\conflict   & $116$\,k    &   $2$\,M      & $0.62$    &    $2.9e\!-\!04$\\
\epinions       & $131$\,k    & $711$\,k\,    & $0.17$       &    $8.2e\!-\!05$\\
\politics   & $138$\,k    & $715$\,k\,    & $0.12$       &    $7.4e\!-\!05$\\
\midrule
\conflict-4   & $1.1$\,M    &   $34.7$\,M      & $0.62$    &    $3.4e\!-\!05$\\
\epinions-4       & $1.1$\,M    & $12.2$\,M    & $0.17$    &    $9.5e\!-\!06$\\
\bottomrule
\end{tabular}
}
\end{table}

\begin{table*}[t] 
  \begin{center}
    \caption{Largest balanced subgraph found by each method for each dataset}
    \vspace{-3mm}
    \label{tab:subgraphs}
      \begin{tabular}{l| rr| rr| rr| rr| rr}
        \toprule
         & \multicolumn{2}{c|}{\tribes} & \multicolumn{2}{c|}{\cloister} & \multicolumn{2}{c|}{\congress} & \multicolumn{2}{c|}{\bitcoin} & \multicolumn{2}{c}{\referendum} \\ \midrule
        method & $|V|$ & $|E|$ & $|V|$ & $|E|$ & $|V|$ & $|E|$ & $|V|$ & $|E|$ & $|V|$ & $|E|$ \\ \midrule
        \ouralgo & 13 & 35 & 10 & 33 & 208 & 452 & 4\,208 & 10\,158 & 8\,944 & 166\,243 \\
        \grasp & 10 & 18 & 6 & 11 & 115 & 145 & 2\,167 & 3\,686 & 5\,425 & 49\,105 \\
        \ggmz & 10 & 21 & 5 & 7 & 153 & 238 & 1\,388 & 1\,683 & 2\,501 & 2\,821 \\
        \eige & 12 & 37 & 8 & 27 & 11 & 16 & 7 & 17 & 132 & 6\,140 \\
        \midrule
        & \multicolumn{2}{c|}{\election} & \multicolumn{2}{c|}{\slashdot} & \multicolumn{2}{c|}{\conflict} & \multicolumn{2}{c|}{\politics} & \multicolumn{2}{c}{\epinions} \\ \midrule
        \ouralgo & 3\,786 & 18\,550 & 42\,205 & 96\,460 & 48\,136 & 356\,204 & 63\,252 & 218\,360 & 62\,010 & 169\,894 \\
        \grasp & 1\,752 & 4\,416 & 23\,289 & 40\,511 & 18\,576 & 82\,726 & 31\,561 & 81\,557 & 28\,189 & 63\,250 \\
        \ggmz & 713 & 771 & 16\,389 & 17\,867 & 6\,137 & 9\,145 & 23\,342 & 37\,098 & 21\,009 & 25\,013 \\
        \eige & 11 & 41 & 35 & 491 & 11 & 28 & 10 & 45 & 6 & 14 \\
        \bottomrule
      \end{tabular}
  \end{center}
\end{table*}

\begin{table*}[t] 
  \begin{center}
    \caption{Solution found by each method for each dataset, compared to the size of the planted balanced subgraph in the graph}
    \vspace{-3mm}
    \label{tab:planted}
    \begin{tabular}{l| rr| rr| rr| rr| rr| rr| rr}
      \toprule
      method & \multicolumn{2}{c|}{\election} & \multicolumn{2}{c|}{\slashdot} & \multicolumn{2}{c|}{\conflict} & \multicolumn{2}{c|}{\politics} & \multicolumn{2}{c|}{\epinions} & \multicolumn{2}{c|}{{\sc Powerlaw-3} } & \multicolumn{2}{c}{{\sc Powerlaw-4} } \\ \midrule
      $|V_p|$& 3500 & \% & 41\,000 & \% & 58\,000 & \% & 69\,000 & \% & 65\,500 & \% & 10\,000 & \% & 10\,000 & \% \\ \midrule
      \ouralgo & 4\,097 & 117\% & 65\,963 & 160\% & 88\,529 & 152\% &  123\,367 & 178\% & 103\,136 & 156\% & 11\,491 & 114\% & 11\,346 & 113\% \\
      \grasp &  1\,072 & 30\% &  9\,640&  23\% & 15\,268 & 26\%  & 11\,253 & 16\% & 9\,313 & 14\% & 4\,858  & 48\% & 5\,834 & 58\% \\
      \ggmz & 1\,952 & 55\% & 14\,320 & 34\% & 12\,671  & 21\% & 25\,202 & 36\% & 16\,944 & 25\% & 9\,341 & 93\% & 9\,344 & 93\% \\
      \bottomrule
    \end{tabular}
  \end{center}
\end{table*}

\subsection{Proposed baselines}
We compare the results of our method to heuristics proposed in the literature for the \mbs, as well as a non-trivial spectral baseline. We now describe these methods.


\smallskip
\noindent
\eige: The spectral approach from \cite{bonchi2019discovering}. We take the dominant eigenvector $\vec{v}$ of the adjacency matrix $A$ of the input graph. For a given threshold $\tau \in \mathbb R_+$, we construct a vector $\vec{x}$ as follows: $\vec{x}_i=sign(\vec{v}_i)$ if $|\vec{v}_i|\geq \tau$, and $\vec{x}_i=0$ otherwise. We try all possible values of $\tau$, that is, all values in $\{|\vec{v}_i| : i=1, \dots, n\}$, where $n$ is the number of nodes in the input graph. Note that $\vec{x}$ defines a partition of the graph into $V_1,V_2$, such that a vertex $i$ is in $V_1$ if and only if $\vec{v}_i=1$, and $i \in V_2$ if and only if $\vec{v}_i=-1$. If the graph is balanced, $\vec{v}$ reveals the perfect partition.


\smallskip
\noindent
\grasp: The \grasp heuristic proposed by Figueiredo and Frota~\cite{figueiredo2014maximum}. The method consists of a construction phase, which greedily builds a balanced partition inspecting the vertices one by one in random order, and a local-search phase. The local-search phase is computationally costly and yielded negligible improvements in our experiments if kept within reasonable running times. Therefore, we only report the results of the construction phase.

\smallskip
\noindent
\ggmz: The heuristic proposed by G{\"u}lpinar et al.~\cite{gulpinar2004extracting}, which functions as follows. First, a minimum spanning tree of the input graph is computed. This tree is then switched --- i.e., a subset of vertices is chosen and all cut edges change sign --- so that all edges become positive. The same switch is applied to the entire graph and a set of vertices that are independent in the negative edge set is returned.

\smallskip

All baselines, as well as our algorithm, are implemented in Python3, using sparse matrix operations when possible.\footnote{Source code: https://github.com/justbruno/finding-balanced-subgraphs}

Other proposals in the literature, such as the work of Chu et al. \cite{chu2016finding}, tackle similar problems, but they tend to find small or imbalanced subgraphs and therefore are not comparable to our methods.

\subsection{\ouralgo implementation details}
As explained in section \ref{sec:several}, our algorithm can be tuned to decide how many vertices are removed at each iteration. Even though better results might be obtained by fine-tuning this value for each problem instance, in our experiments we take a simple approach to balance quality and efficiency. In the small datasets --- \tribes, \cloister~ and \congress\ --- we remove one vertex at a time. In the rest of the datasets, we remove at most 100. Note that the set of removed vertices can be smaller if we cannot find 100 independent vertices.

In the first stage of the algorithm, whenever the removal of a vertex subset results in various connected components, we discard all but the largest. In all our experiments, all other components are very small, comprised of a handful of vertices at most, and are not worth analysing further.

In all experiments with large data sets ($\geq 82k$ vertices) we randomly sample $1000$ subgraphs of approximately 200 vertices and process them as explained in section \ref{sec:scaling}.
\subsection{Finding balanced subgraphs}
The main purpose of the algorithm described in this paper is to find large, preferably dense, balanced subgraphs in signed networks. Therefore, we first evaluate the ability of our method and the proposed baselines to accomplish this goal. We run all algorithms on all datasets, and measure the size of the obtained subgraph in both vertices ($|V|$) and edges ($|E|$). We report the results in Table \ref{tab:subgraphs}. Since the algorithms are sensitive to the choices made in their initial stages, we run them ten times and report the maximum values achieved in both vertex and edge count. 
  \ouralgo obtains larger results in all datasets, significantly so in most of them.

Although our algorithm finds balanced subgraphs of large size, we cannot know how far off we are from an optimal solution. 
To evaluate this, we do a further experiment where we create alternative versions of our datasets with planted balanced subgraphs of selected size. In order to do this, we make a randomly selected part of the graph balanced, by switching the edge signs accordingly, while randomizing the signs of all the other edges. As an additional performance test, we also generated some power-law graphs, using the Barab{\'a}si-Albert model \cite{barabasi1999emergence} and planted balanced subgraphs in them using the same procedure. The first power-law graph has $20\,000$ nodes and $59\,991$ edges 
($m=3$),\footnote{The parameter $m$ here refers to the one used in the Barab{\'a}si-Albert model.} while the second one has $20\,000$ and $79\,984$ edges ($m=4$). We experiment with planted balanced subgraphs of approximately half the size of total graph. The results can be seen in Table \ref{tab:planted}. $|V_p|$ denotes the size of the planted balanced subgraph.  {\sc Powerlaw-3} and {\sc Powerlaw-4} are Barab{\'a}si-Albert graphs with $m=3$ and $m=4$ respectively. We also report the percentage of the planted subgraph that was recovered by the algorithm (notice that it may be bigger than the planted subgraph, due to the added random noise). 

We observe that in all cases, our algorithm manages to return a balanced subgraph of size at least the size of the planted balanced subgraph, while the baselines fail to do so, for all datasets. 

\begin{figure*}[t]
  \begin{tabular}{cccc}    
  \includegraphics[width=.24\textwidth]{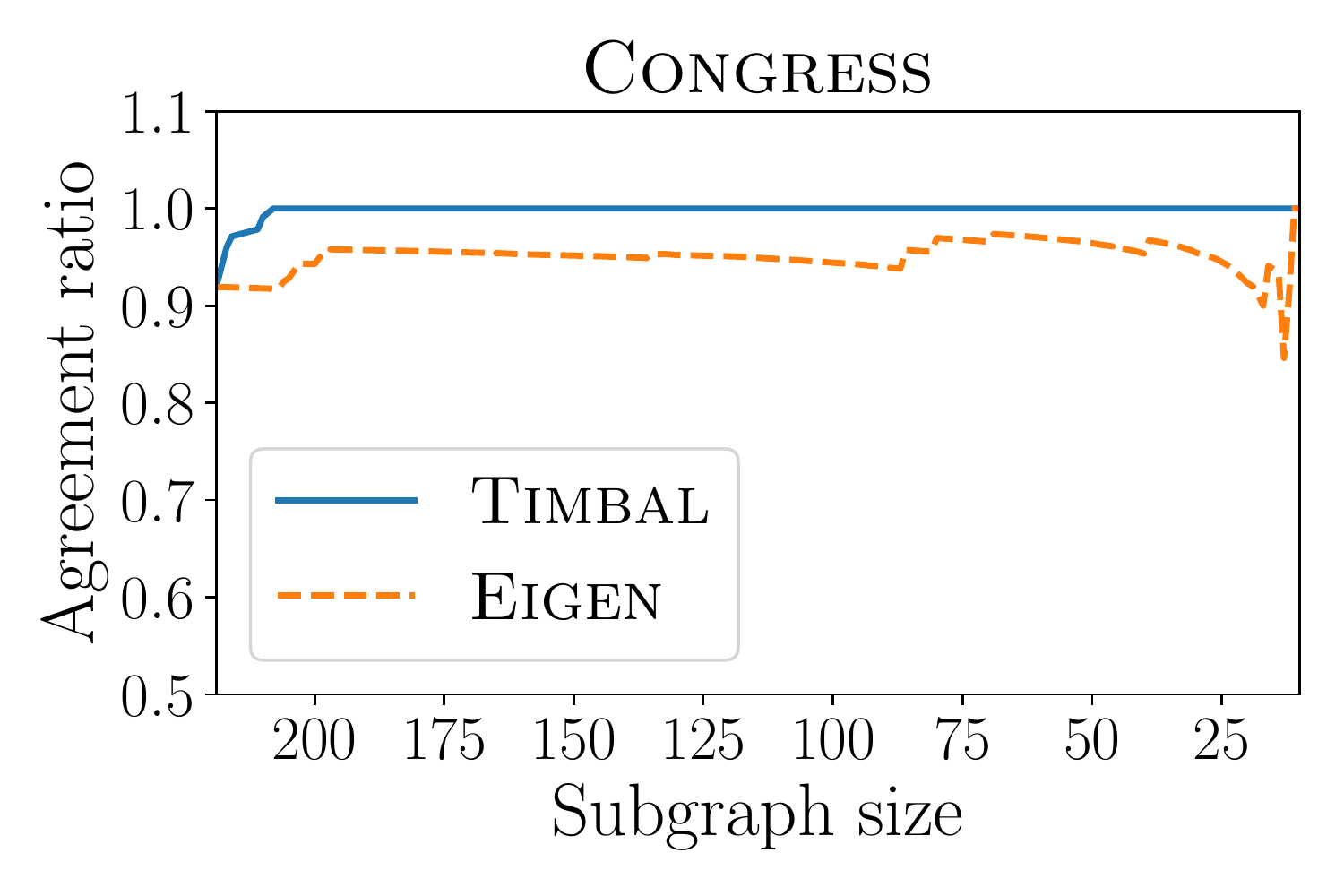} &   \includegraphics[width=.24\textwidth]{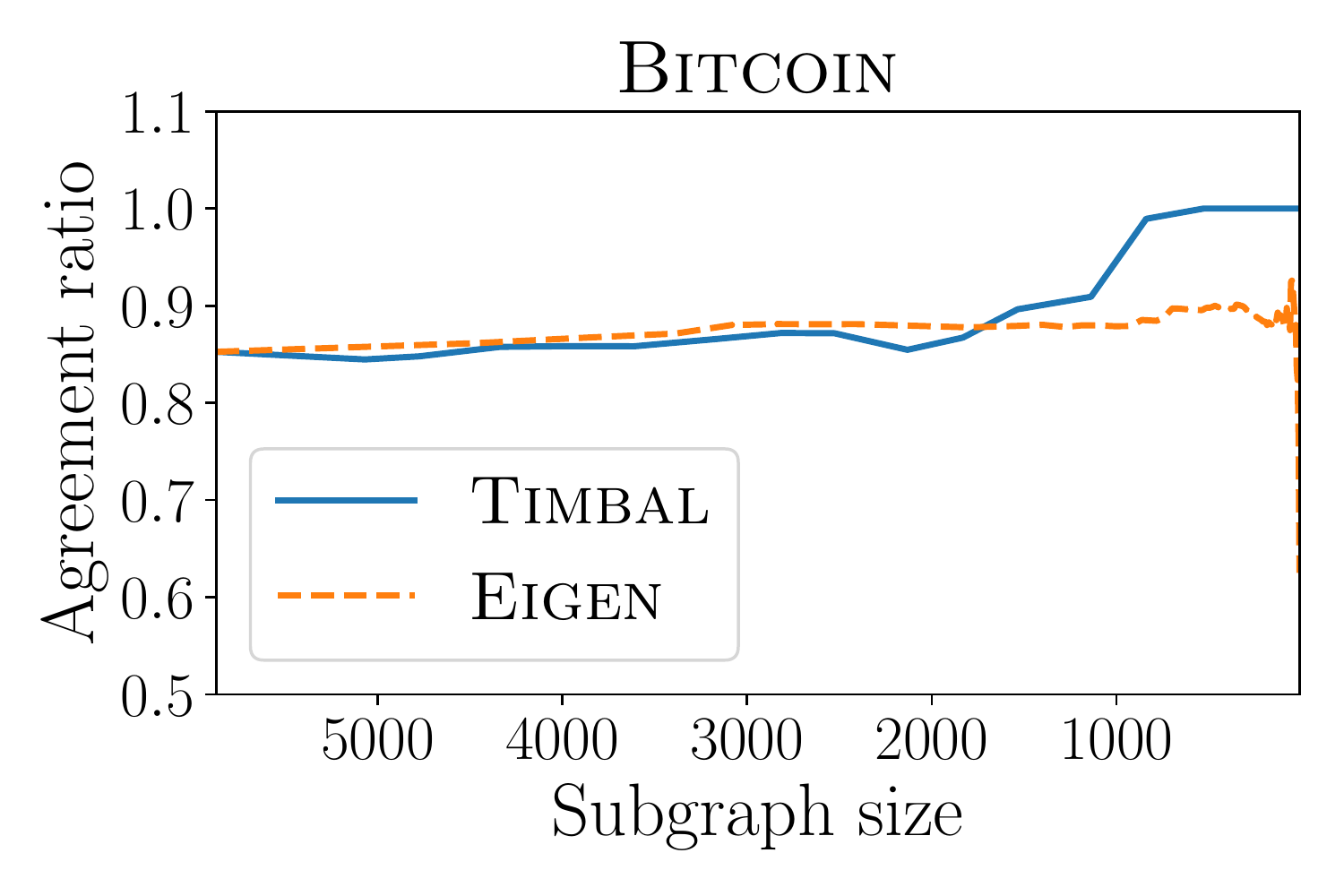} &  \includegraphics[width=.24\textwidth]{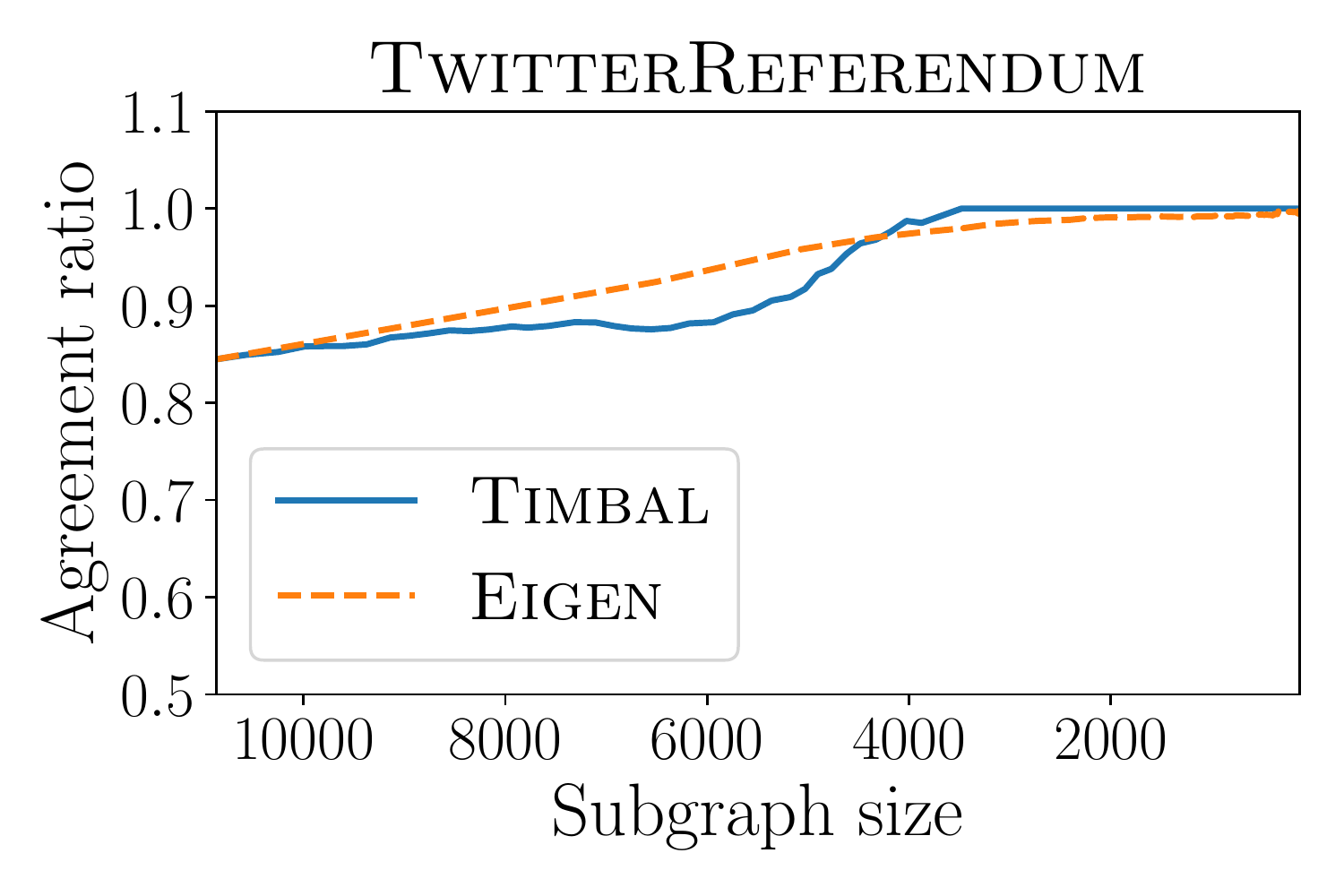} &   \includegraphics[width=.24\textwidth]{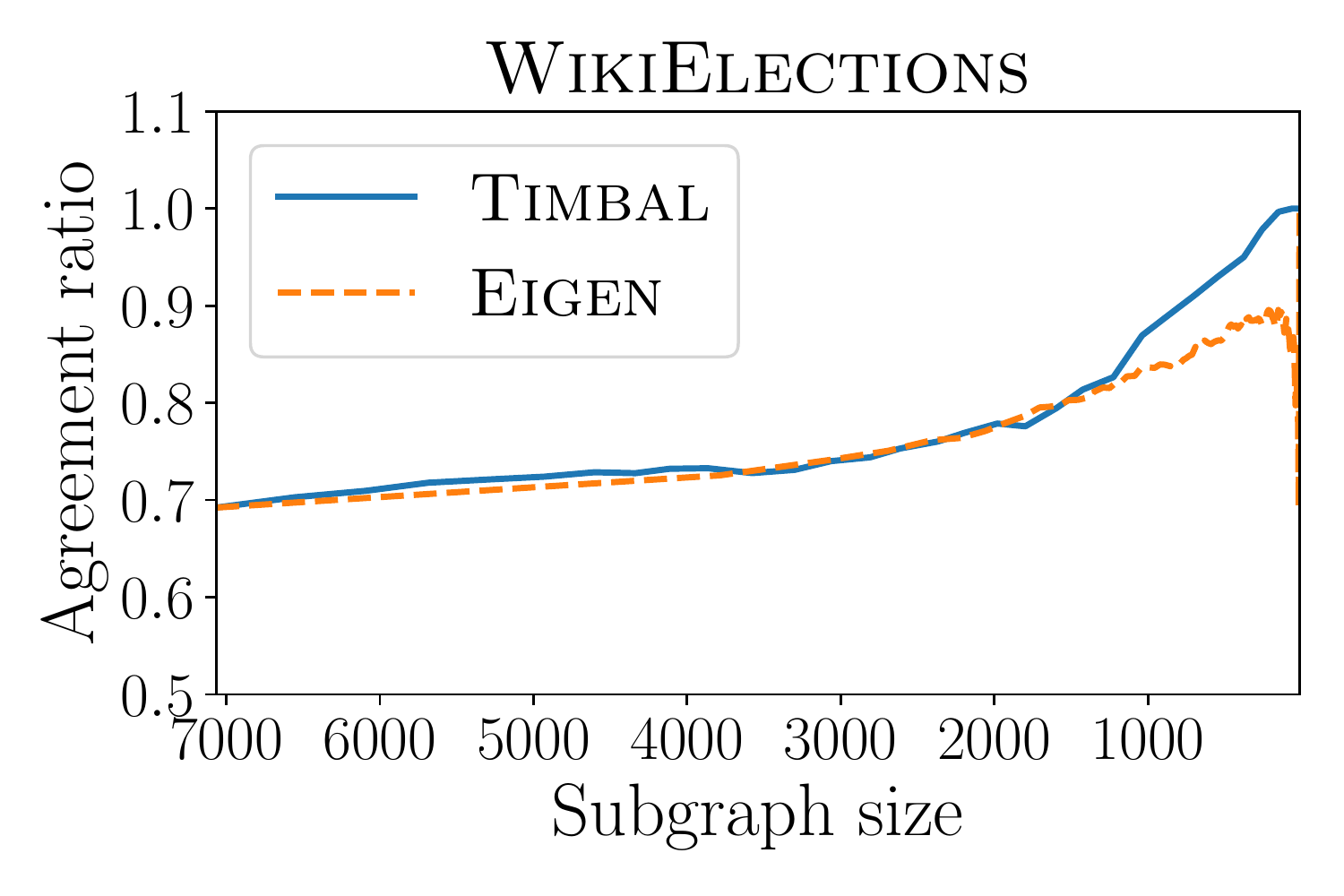}  \\
 \includegraphics[width=.24\textwidth]{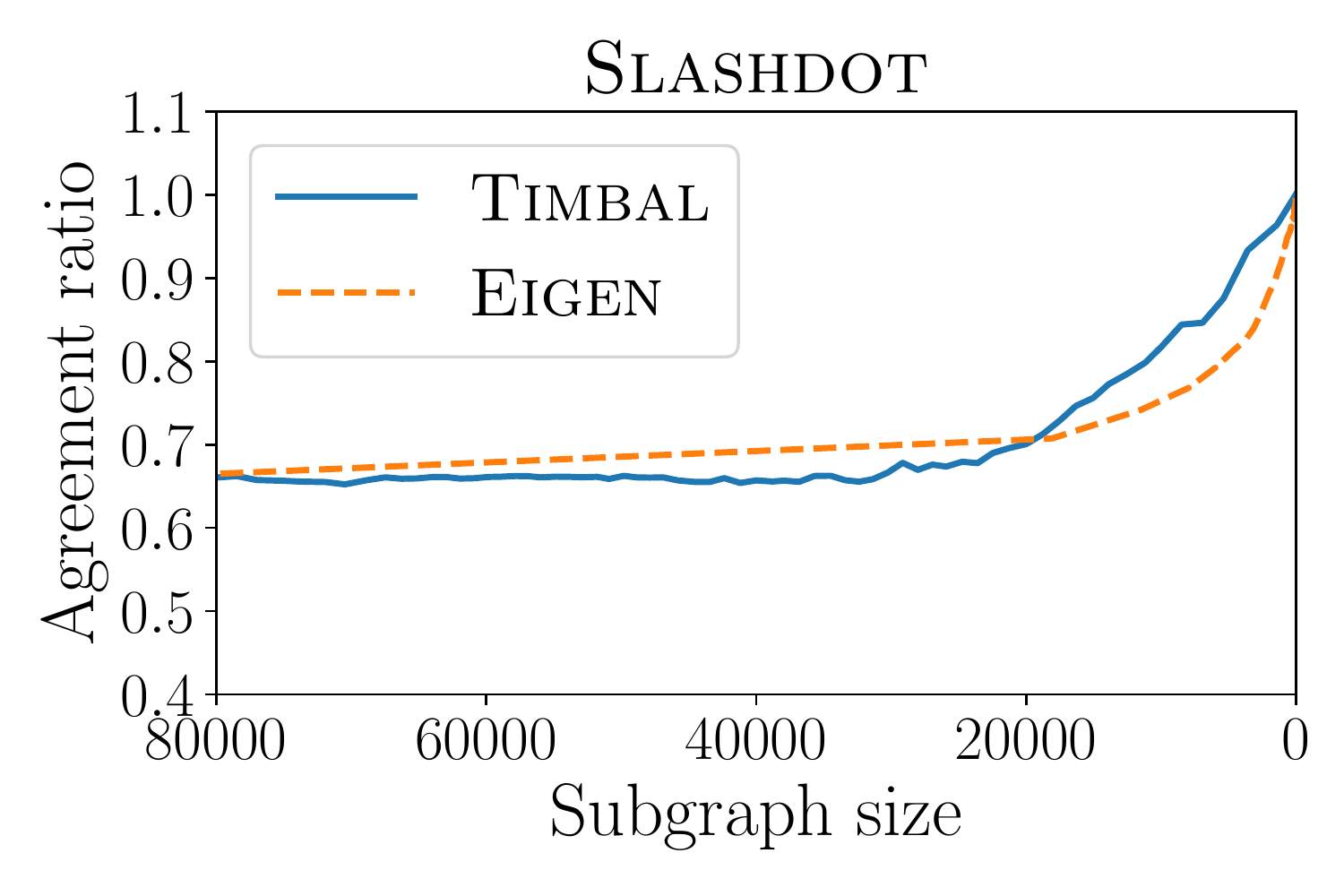} &   \includegraphics[width=.24\textwidth]{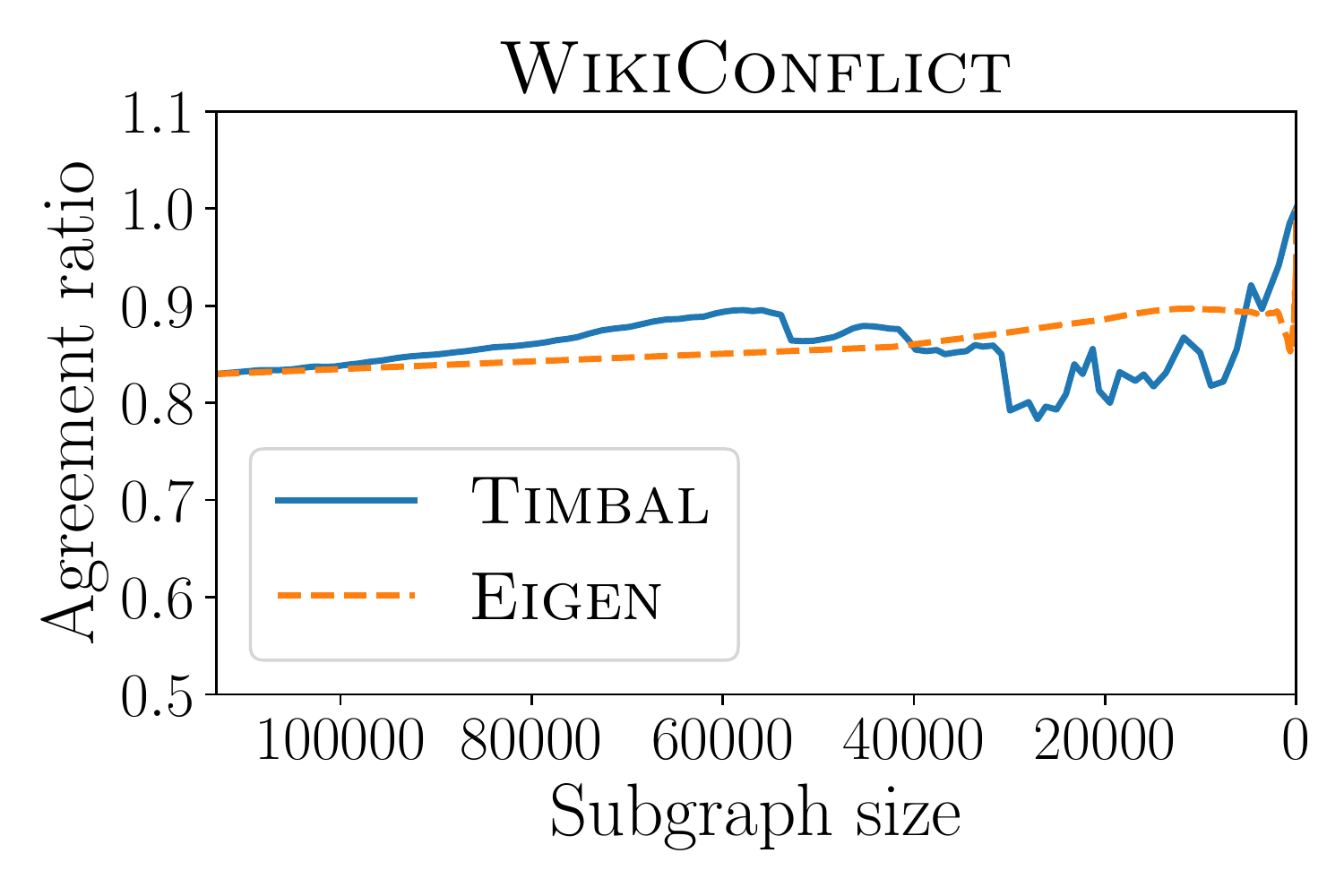} & \includegraphics[width=.24\textwidth]{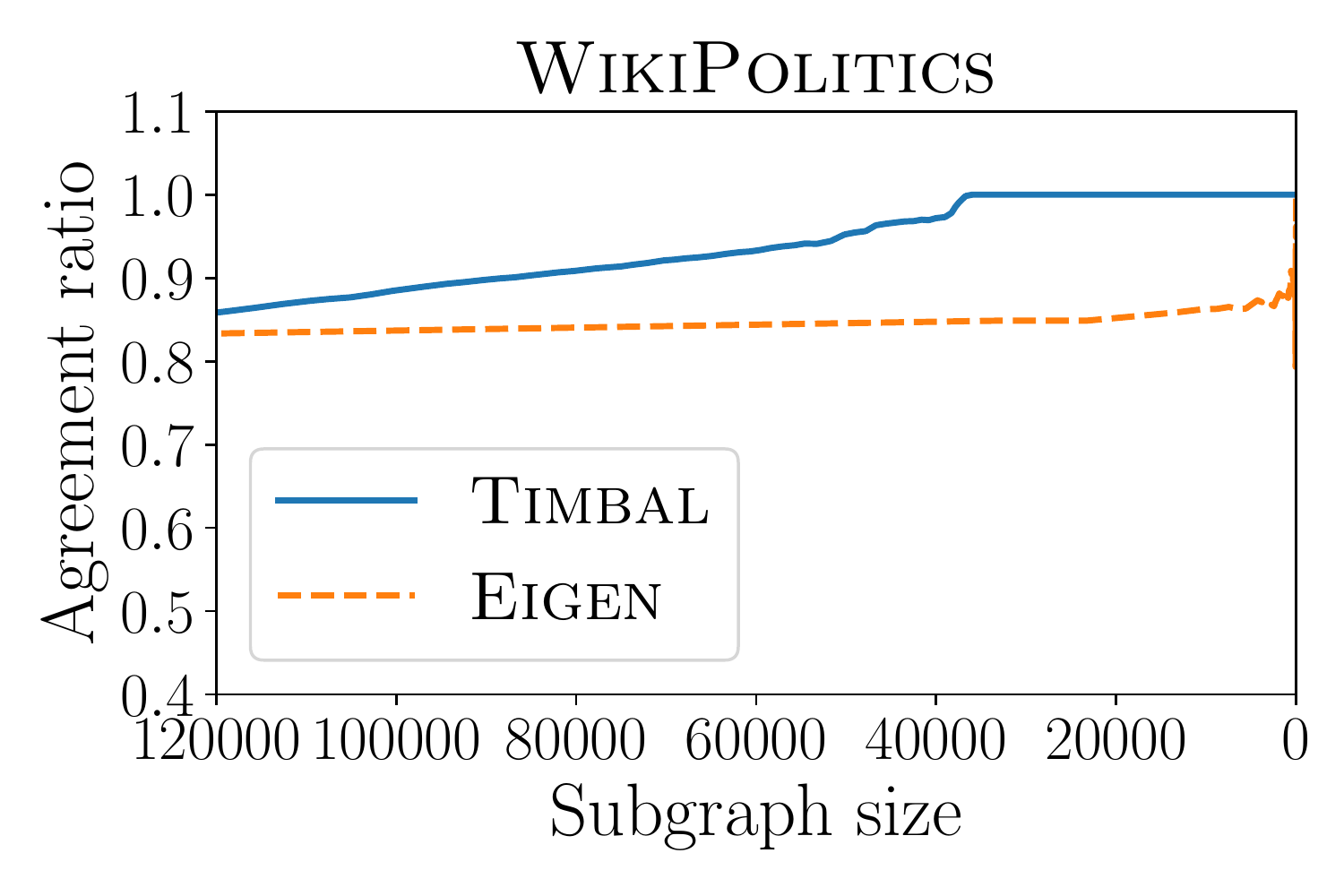} &   \includegraphics[width=.24\textwidth]{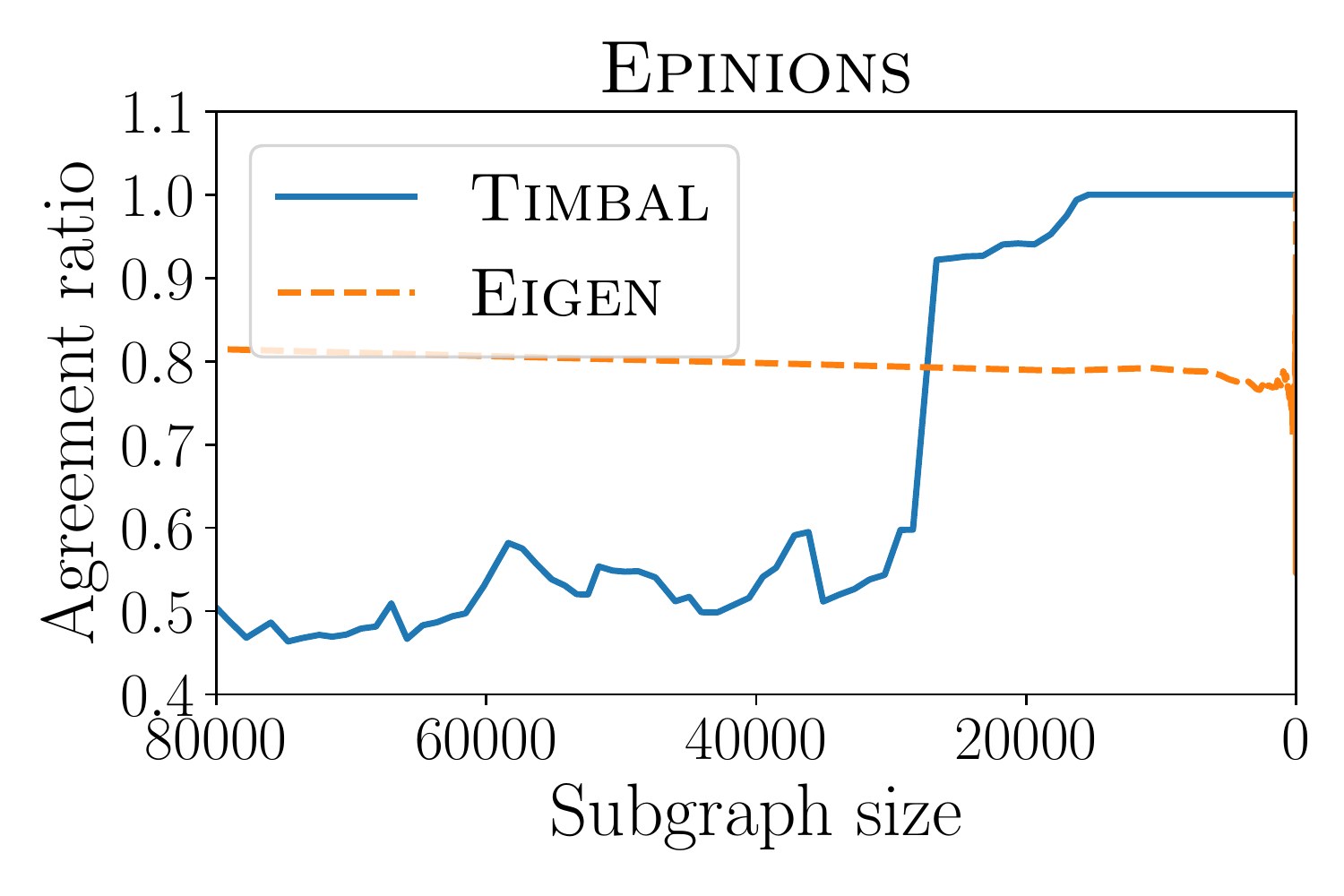}  \\
\end{tabular}
  \vspace{-3mm}
  \caption{Edge agreement ratio of all subgraphs visited by \ouralgo and \eige.}
    \label{fig:visits}
\end{figure*}

\begin{figure*}[t]
  \begin{tabular}{cccc}    
  \includegraphics[width=.24\textwidth]{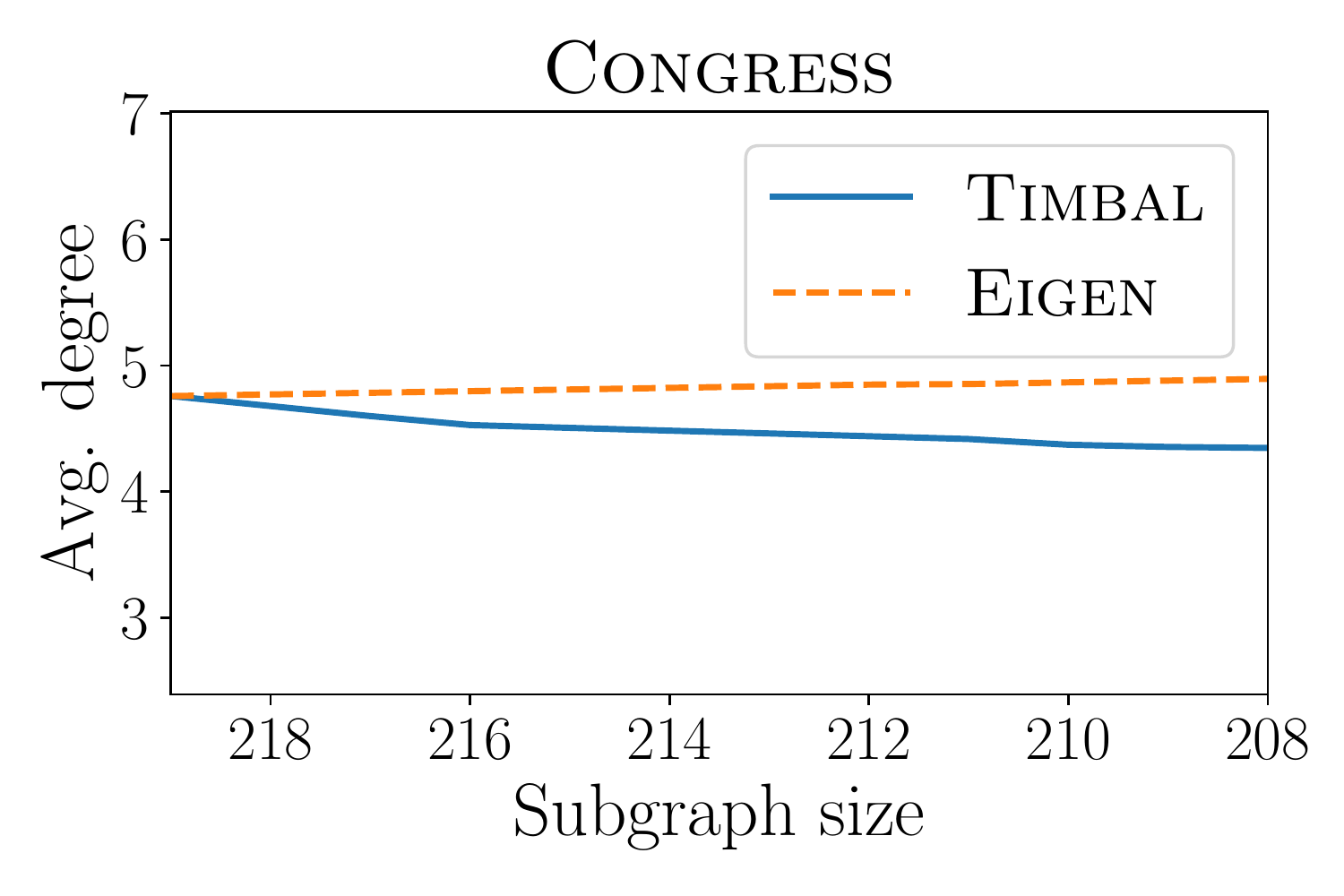} &   \includegraphics[width=.24\textwidth]{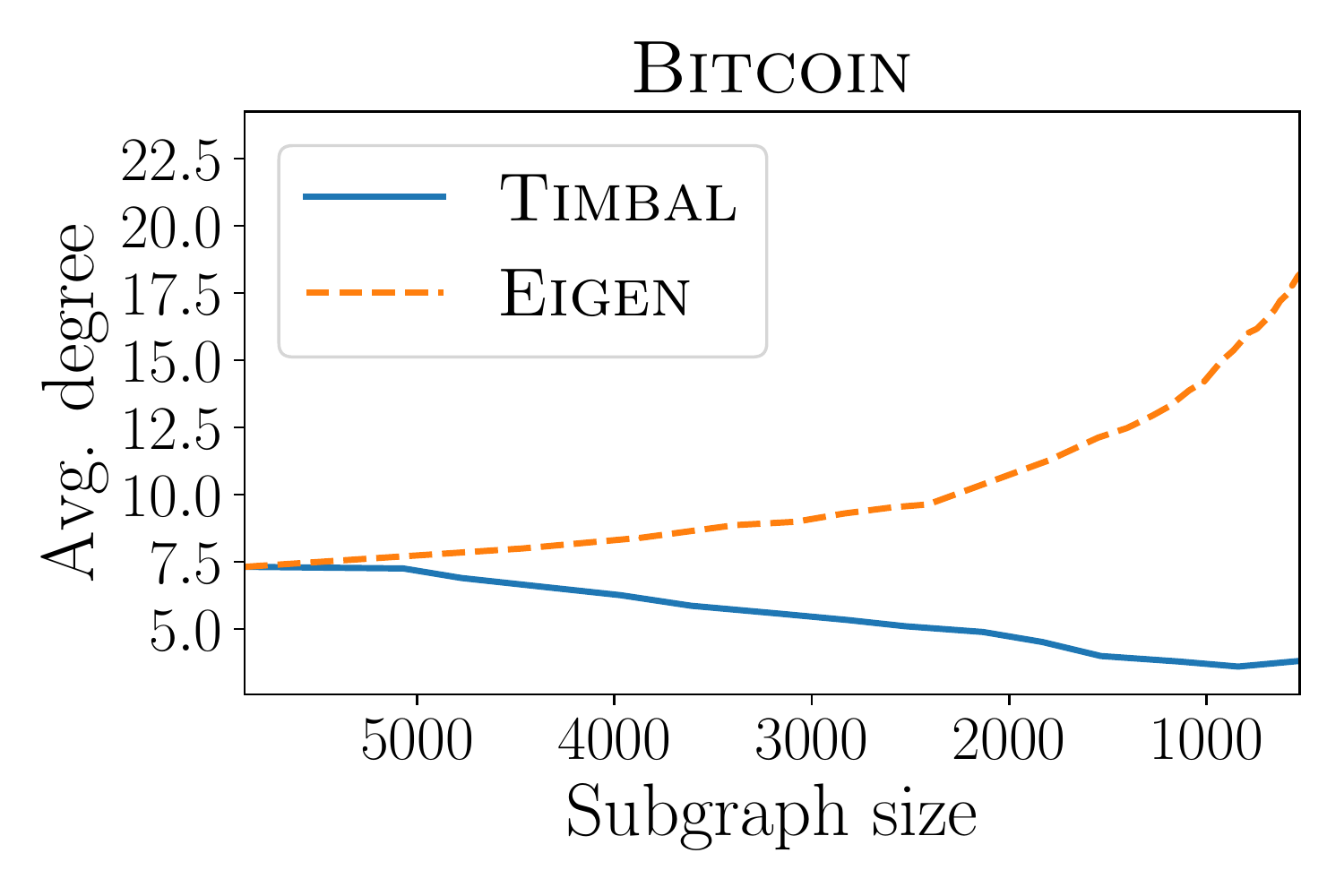} &  \includegraphics[width=.24\textwidth]{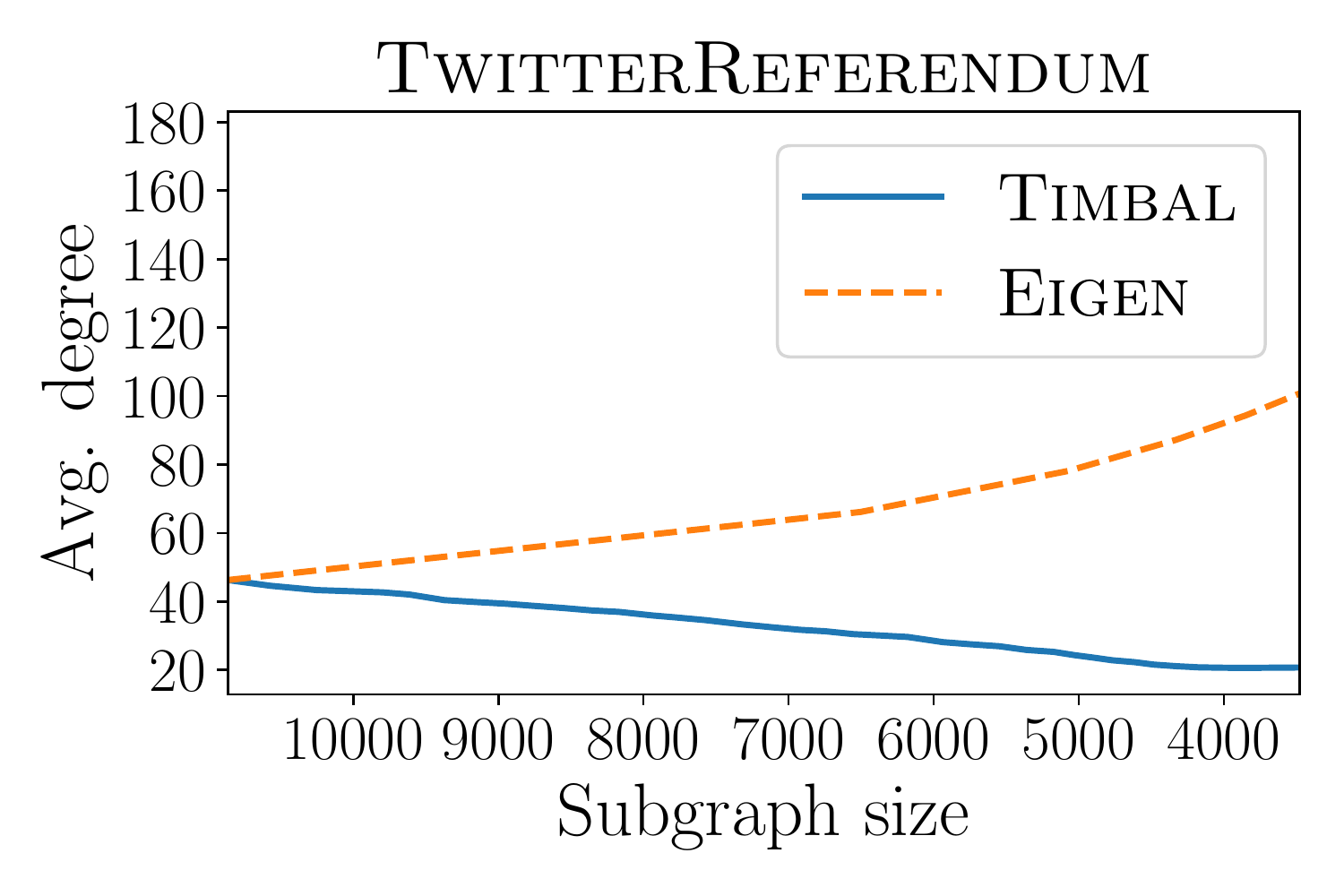} &   \includegraphics[width=.24\textwidth]{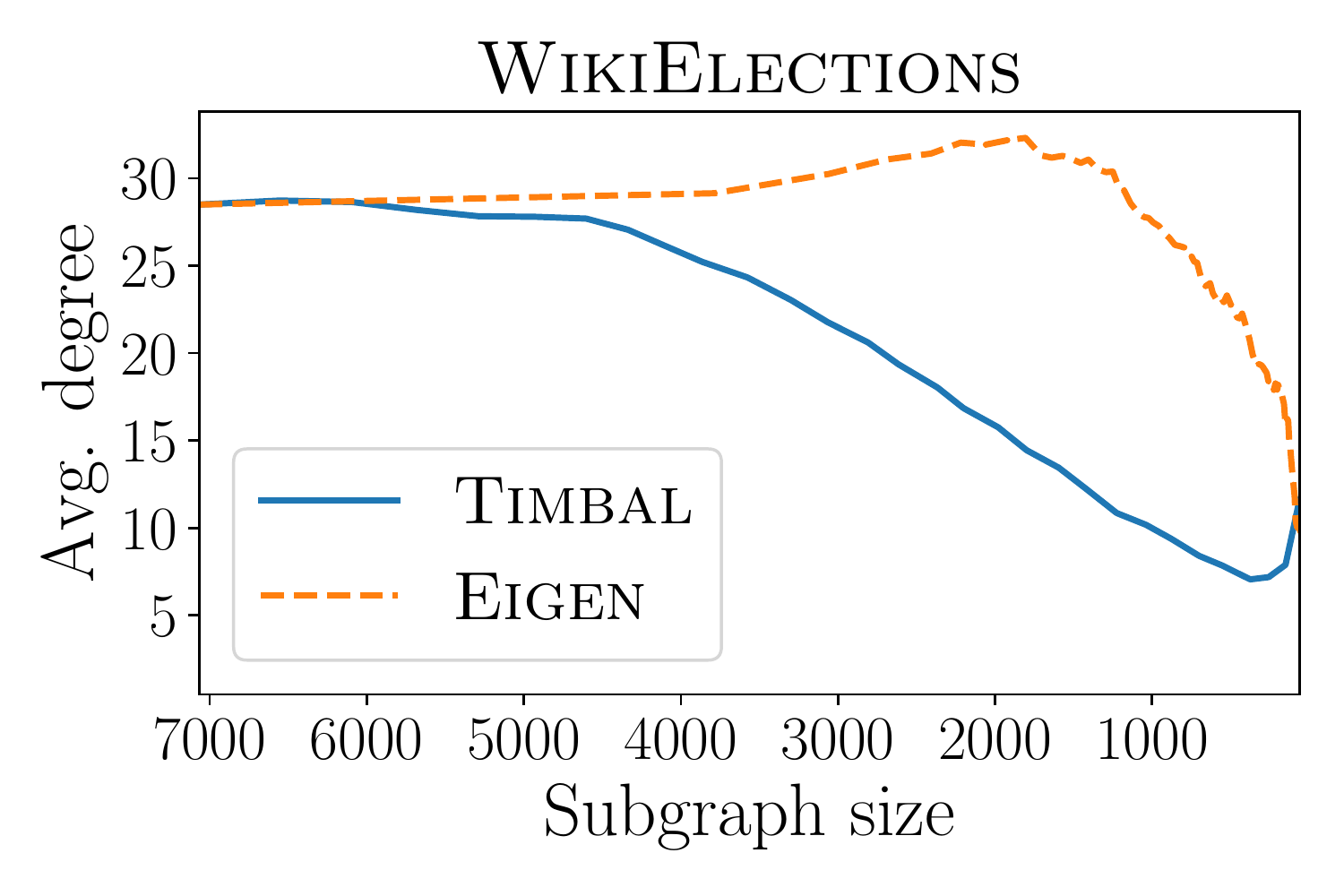}  \\
 \includegraphics[width=.24\textwidth]{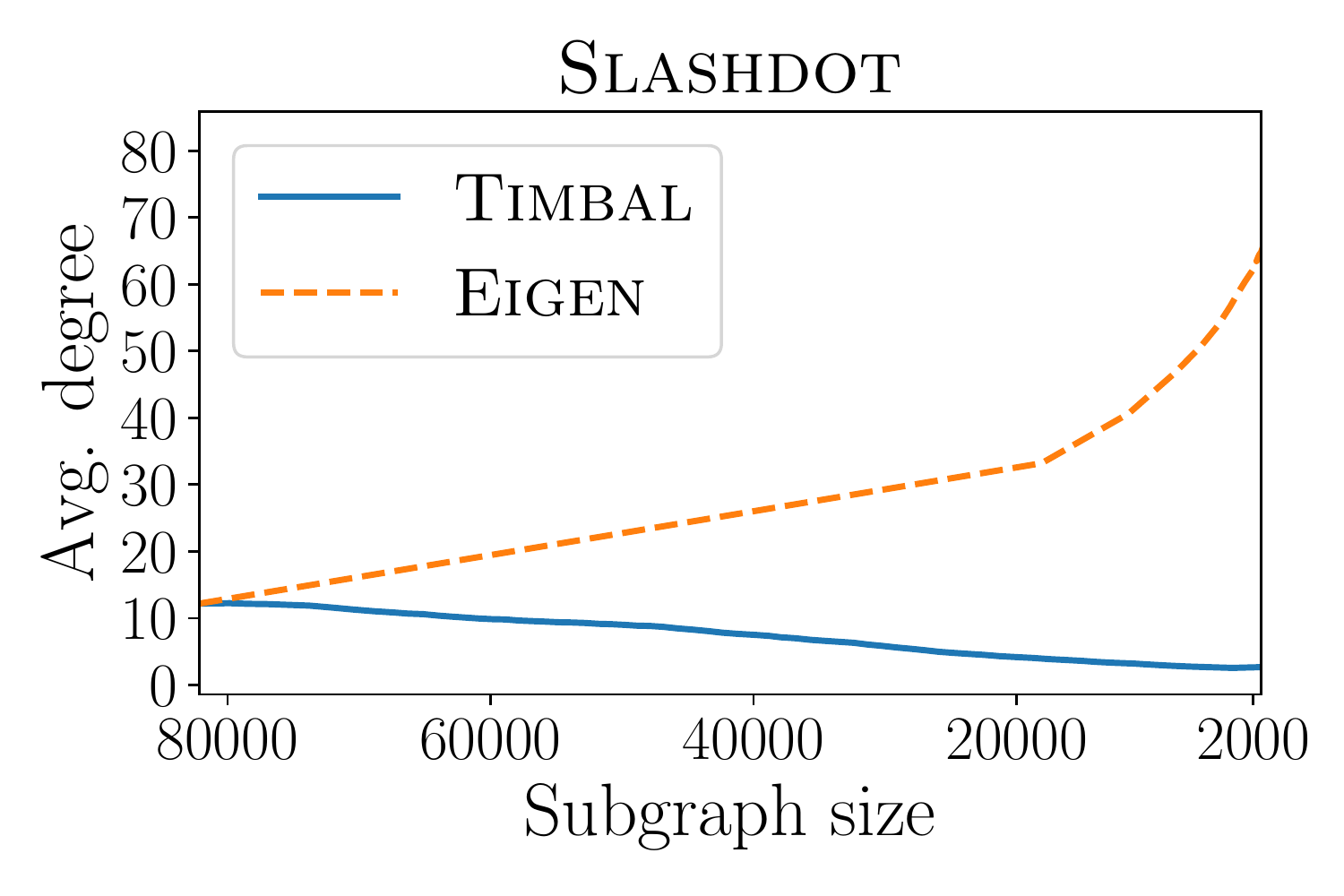} &   \includegraphics[width=.24\textwidth]{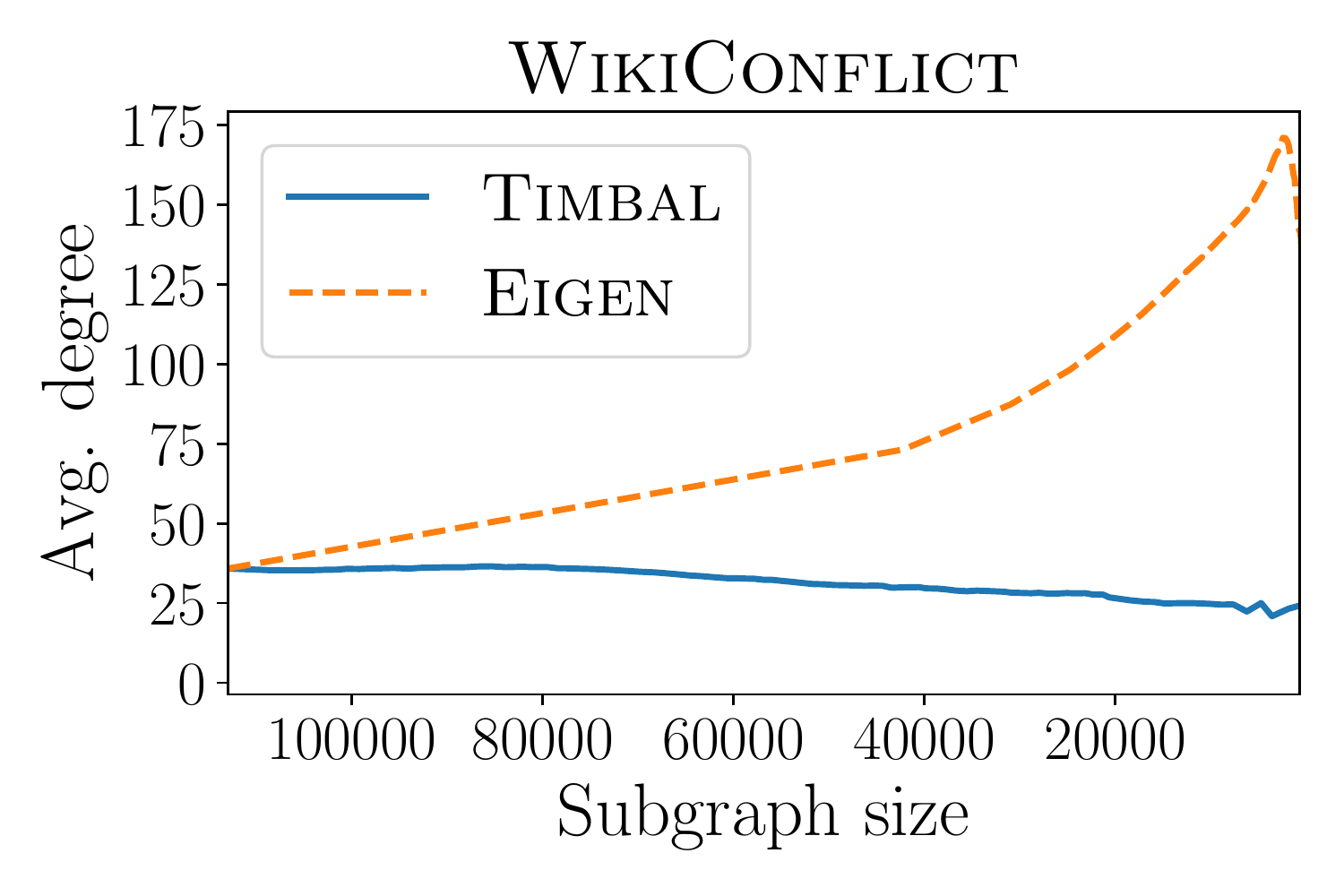} & \includegraphics[width=.24\textwidth]{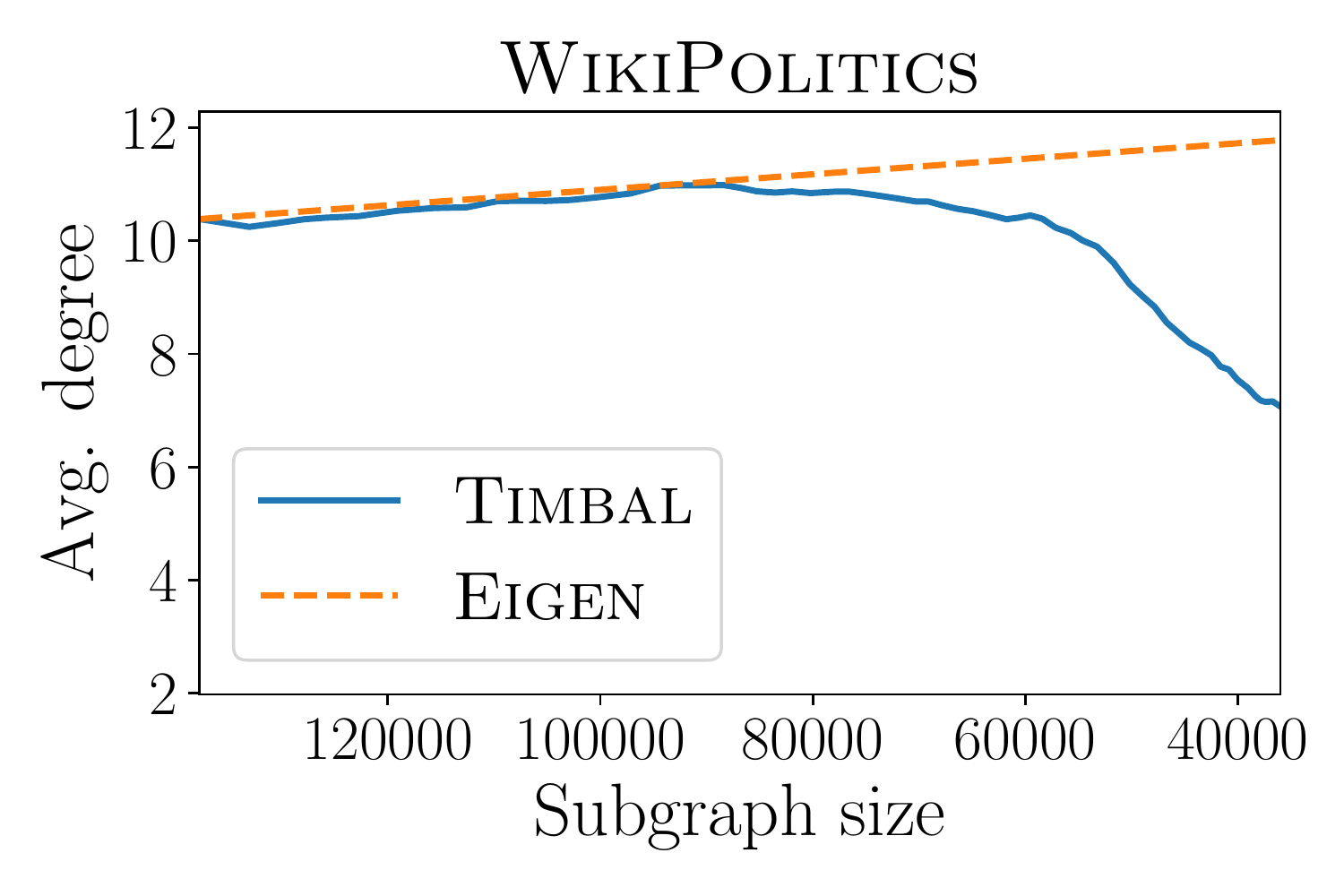} &   \includegraphics[width=.24\textwidth]{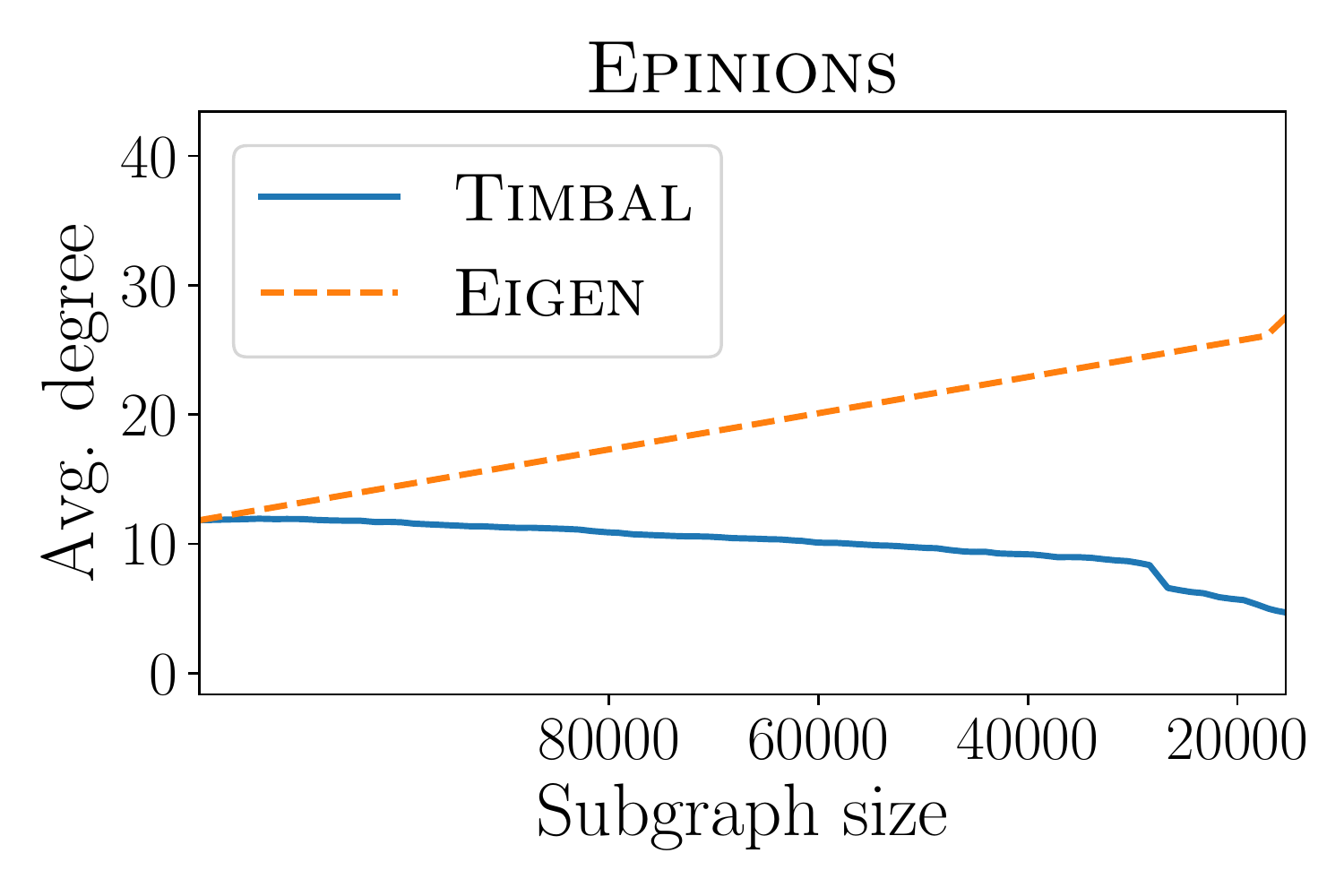}  \\
\end{tabular}
  \vspace{-3mm}
  \caption{Average degree of all subgraphs visited by \ouralgo and \eige.}
    \label{fig:density}
\end{figure*}

\begin{table}[t] 
  \begin{center}
  \begin{small}
    \caption{Mean running times in seconds for each algorithm on the larger datasets, with the corresponding variance reported in the brackets}\vspace{-3mm}
    \label{tab:time}
      \begin{tabular}{l rrrrrrrr}
        \toprule
         & \multicolumn{1}{c}{\slashdot} & \multicolumn{1}{c}{\conflict} & \multicolumn{1}{c}{\politics} & \multicolumn{1}{c}{\epinions} \\ 
        \midrule
        \ouralgo & 117\,(5.85) & 159\,(36.35) & 210\,(33.94) & 244\,(31.11) \\
        \grasp & 59\,(0.87) & 105\,\,\,(3.07) & 154\,(13.27) & 116\,\,\,(0.69) \\
        \ggmz & 318\,(0.20) & 461\,\,\,(0.56) & 528\,\,\,(1.36) & 670\,\,\,(0.32) \\
        \eige & 14 $\qquad$\, & 82 $\qquad$\,\, & 25 $\qquad$\,\, & 48 $\qquad$\,\, \\
        \bottomrule
      \end{tabular}
  \end{small}
  \end{center}
\end{table}

\begin{table}[t] 
  \begin{small}
  \begin{center}
    \caption{Running times in seconds for our algorithm on the artificially augmented datasets. We report the size of each augmented dataset ($|V|$, $|E|$), as well as the size of the solutions ($|V_S|$, $|E_S|$).}\vspace{-3mm}
    \label{tab:scala}
      \begin{tabular}{l rrrrr }
        \toprule
        Dataset & $|V|$ & $|E|$ & Time (s) & $|V_S|$ &  $|E_S|$ \\
        \midrule
        \conflict-1 & 233\,434 & 6.1\,M & 260 &  67\,962 & 718\,455  \\
        \conflict-2 & 350\,151 & 10.1\,M & 431 & 75\,024  & 759\,280 \\
        \conflict-3 & 583\,585 & 18.3\,M & 798 & 99\,506  & 808\,804  \\
        \conflict-4 & 1.05\,M  & 34.7\,M & 2\,059 & 152\,789 &  964\,446 \\
        \epinions-1 & 263\,160 & 2.1\,M & 306 &  1119\,73 & 348\,502  \\
        \epinions-2 & 394\,740 & 3.6\,M &  590 & 153\,419 &  481\,378 \\
        \epinions-3 & 657\,900 & 6.5\,M & 1\,776 & 231\,709 & 695\,614 \\
        \epinions-4 & 1.1\,M & 12\,M & 5\,628 & 385\,478  & 1\,081\,607  \\
        \bottomrule
      \end{tabular}
  \end{center}
  \end{small}
\end{table}

\begin{figure}[t]  
  \includegraphics[width=.9\columnwidth, trim=0cm 1.25cm 0cm .85cm,clip]{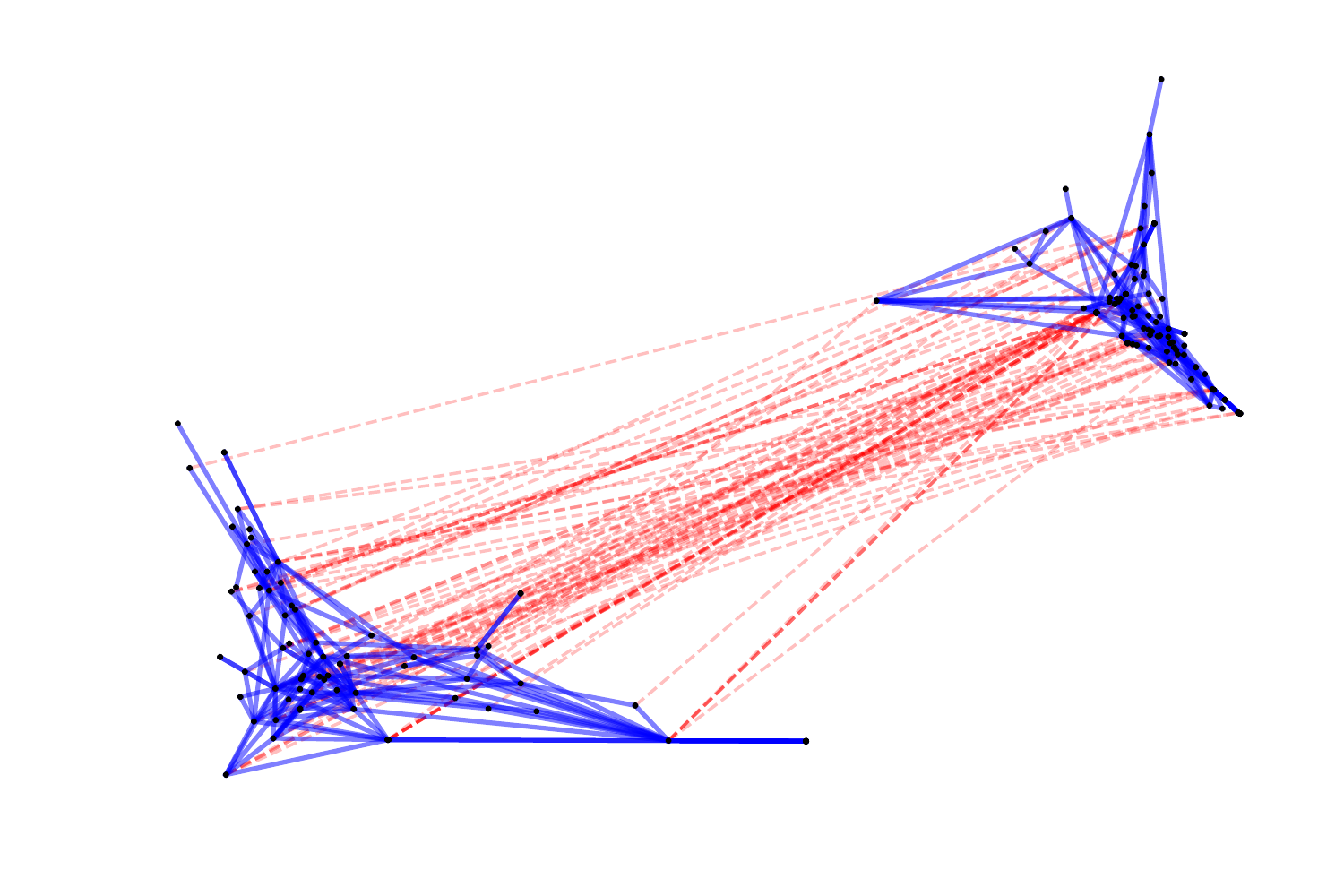}
  \caption{Visualization of the result for \congress. Solid blue edges are positive, while dashed, red ones are negative.}
  \label{fig:congress}
\end{figure}

\begin{figure*}[h!]
   \includegraphics[width=0.8\textwidth, trim=0cm 1.25cm 0cm .5cm,clip]{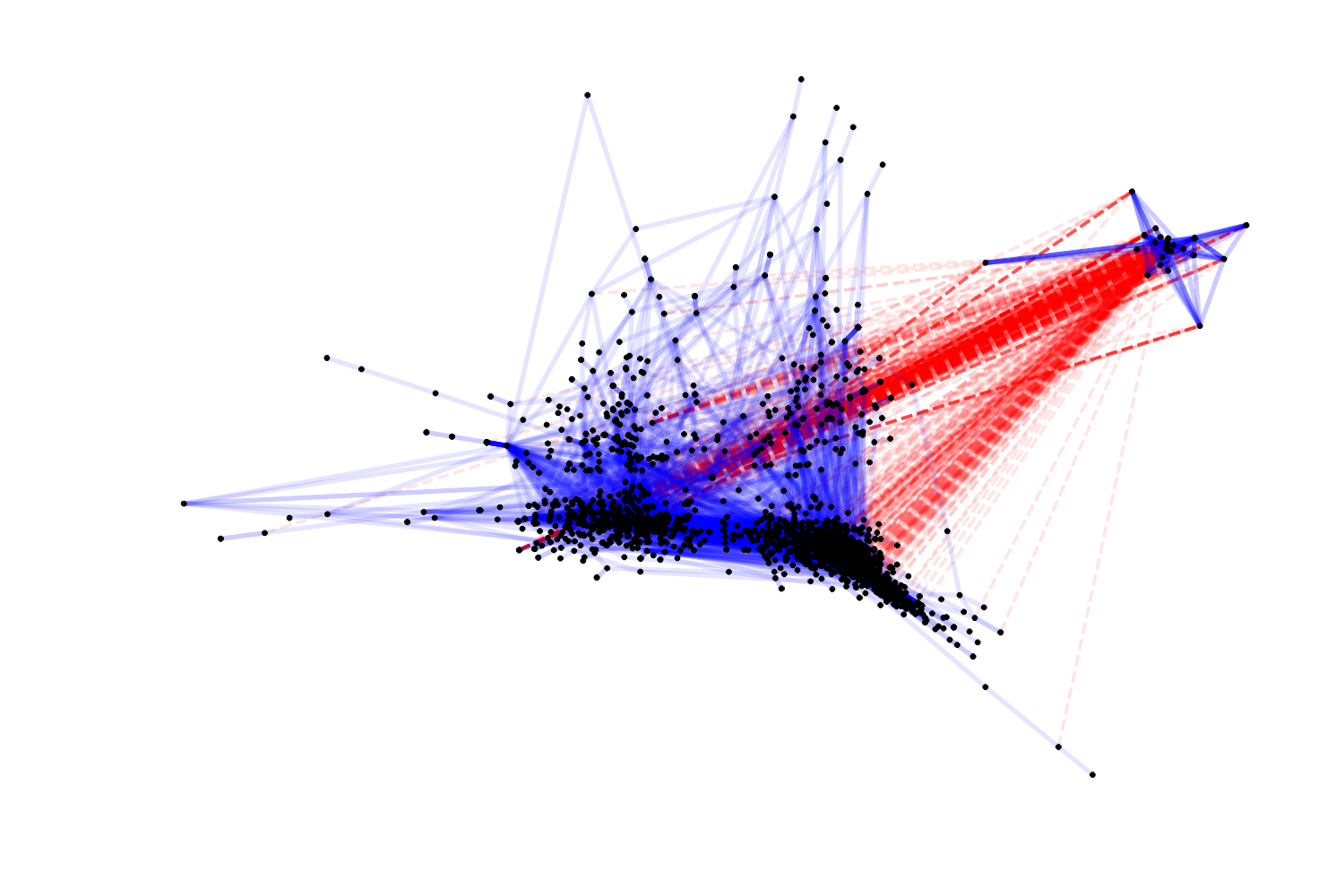}
  \caption{Visualization of the result for \bitcoin. Solid blue edges are positive, while dashed, red ones are negative.}
  \label{fig:bitcoin}
\end{figure*}

\subsection{Trading off balance and graph size}
Even though our method is intended to find balanced subgraphs, in some applications a small number of mistakes might be acceptable if that means we can find a larger, denser subgraph. An advantage of our algorithm is that the first stage produces a sequence of vertices to remove so that the graph becomes increasingly balanced. Therefore, we can inspect the subgraph obtained at every step of the process and keep the one that best suits our purposes.

To evaluate the extent to which we can benefit from this, we inspected all the graphs generated during the first stage of \ouralgo and measured their quality in terms of balance. We compared its performance to \eige, which is the only baseline that can produce a removal sequence.\footnote{\grasp inspects the vertices at a random order, and \ggmz finds an independent set, the complement of which does not have an intrinsic order.} Figure \ref{fig:visits} shows the results. For each visited graph, we indicate its size on the x axis and plot its corresponding \textit{edge agreement ratio}, which we define as follows. Given a graph $\graph$ with adjacency matrix $A$, let $\vec{x}$ be the indicator vector obtained by taking the signs of the dominant eigenvector of $A$. Then the edge agreement ratio of $\graph$ is
\begin{align}
  \frac{\vec{x}^TA\vec{x}}{\|A\|_F^2},
\end{align}
where $\|A\|_F$ denotes the Frobenius norm. Thus, the edge agreement ratio quantifies the proportion of edges in the graph that agree with the eigenvector-based partition. In a balanced graph, this quantity is 1. Note that in most cases, \ouralgo provides significantly better subgraphs than \eige throughout the entire process.
 However, \ouralgo is able to locate a large balanced subgraph at some step, whereas \eige cannot. Recall that all plots discussed in this section correspond to the first stage of \ouralgo.

Additionally, Figure~\ref{fig:visits} provides some insights about the behaviour of \ouralgo. For instance, in various cases the edge agreement ratio increases sharply at some point. This suggests that once the graph is close to balance, our method can quickly find which vertices to remove to achieve perfect balance.

Figure~\ref{fig:density} shows analogous results, but plotting average degree instead of edge agreement ratio. As expected, \eige achieves higher density than \ouralgo. This can be explained by the fact that \eige can be seen as optimizing the following objective \cite{bonchi2019discovering}:
\[
\max_{\vec{x}\in \{-1,0,1\}^n}\frac{\vec{x}^TA\vec{x}}{\vec{x}^T\vec{x}}.
  \]
However, this comes at a noticeable cost in edge agreement ratio. \ouralgo, on the contrary, tries to achieve balance, which results in reduced density. Nevertheless, Figure~\ref{fig:density} also illustrates how one can trade off balance for density in the different stages of the execution of \ouralgo. If desired, one can take the graph visited at some iteration of the first stage of the algorithm and then execute the second stage, adding the vertices that agree with the current best partition, even if perfect balance cannot be achieved.

\subsection{Running times}
We report the running times of the algorithms on the larger datasets in Table \ref{tab:time}. The experiments are executed on a machine equipped with an Intel Xeon E5-2670 with 24 cores and 256 GB of RAM. We run each algorithm ten times on each dataset and report averages, as well as the variance of the running times. \ouralgo is executed removing 100 vertices at each iteration.

We only report results for the larger subgraphs, as the running times are similar for all algorithms on the rest. 
As expected, \eige is the fastest method, as its running time is dominated by the computation of the dominant eigenvector of the adjacency matrix of the input graph. However, its results in terms of balanced subgraph size are poor. Among the competitive methods, \ouralgo arguably achieves the best combination of quality and running time ratio.

\subsection{Scalability}
In order to assess the scalability of our algorithm, we augment two of the larger datasets 
(i.e., \conflict\ and \epinions)
by artificially injecting vertices with a number of randomly-connected edges equal to the average degree of the original network, while maintaining $\rho_-$ (i.e., the ratio of negative edges). The largest obtained datasets are comprised of about 1.1 million vertices and 34 million edges in the case of \conflict, and 1.1 million vertices and 12 million edges in the case of \epinions.


We execute our algorithm on these datasets. For each dataset, we make five iterations and report the average running time, as well as the size of the best solution found during the iterations. To improve running times, we dynamically set the number of vertices to remove in the $i$-th iteration to $n_i/100$, where $n_i$ is the size of the subgraph at iteration $i$. Notice that the graphs found by the algorithm are large. The results are shown in Table~\ref{tab:scala}.

\subsection{Examples}

In order to gain insight on the results of our methods, we plot the discovered balanced subgraphs for two of the datasets. The graphs discussed in this section are those found in the first stage of \ouralgo.

\spara{\congress:} A balanced subgraph comprised of 208 vertices was found. The first subset in the partition had 95 vertices and 173 edges, while the second had 113 vertices and 199 edges. There are 80 edges between the two sets. Note that this dataset contains 218 vertices in total. This result reveals that the individuals represented by the vertices of this graph are very polarized, as they can be perfectly partitioned by removing just nine of them.
The result is depicted in Figure~\ref{fig:congress}.

\spara{\bitcoin:} A balanced subgraph comprised of 3254 vertices was found. The first subset in the partition had 2986 vertices and 6504 edges, while the second had 268 vertices and 173 edges. There are 807 edges between the two sets. In this case, notice that the second set is more densely connected to the other set than within itself. This reveals that the graph contains a large community of affine users towards which many users feel negatively. The result is depicted in Figure~\ref{fig:bitcoin}.

\section{Conclusions and future work}
\label{sec:conclusion}

In this paper we presented a novel, efficient algorithm for finding large balanced subgraphs in signed networks. By relying on signed spectral theory and a novel bound for perturbations of the graph Laplacian, 
we derived an efficient implementation. Through a wide variety of experiments on real-world and synthetic data we showed that our method achieves better results, in shorter or comparable time, than state-of-the-art methods. We tested scalability on graphs of up to 34\,M edges.

Our work leaves several open avenues of further inquiry. First of all, it would be interesting to study the problem of optimally choosing a constrained subset of vertices to decrease the smallest Laplacian eigenvalue as much as possible. To the best of our knowledge, this problem has not been considered before in the literature. Second, we would like to carry out a thorough analysis of the impact of the number of removed vertices on the quality of the results and running time of the algorithm. A more interesting question to answer is whether we can determine this number optimally at each iteration. Finally, can we further improve the scalability of our algorithm? How efficiently can we find a large balanced subgraph in massive networks?

\vfill
\spara{Acknowledgments.}
This work was supported by 
three Academy of Finland projects (286211, 313927, 317085), 
the EC H2020RIA project ``SoBigData++'' (871042), 
and
the Wallenberg AI, Autonomous Systems and Software Program (WASP) funded by Knut and Alice Wallenberg Foundation.

\bibliographystyle{ACM-Reference-Format}
\flushend
\bibliography{citations}


\begin{thebibliography}{36}


\ifx \showCODEN    \undefined \def \showCODEN     #1{\unskip}     \fi
\ifx \showDOI      \undefined \def \showDOI       #1{#1}\fi
\ifx \showISBNx    \undefined \def \showISBNx     #1{\unskip}     \fi
\ifx \showISBNxiii \undefined \def \showISBNxiii  #1{\unskip}     \fi
\ifx \showISSN     \undefined \def \showISSN      #1{\unskip}     \fi
\ifx \showLCCN     \undefined \def \showLCCN      #1{\unskip}     \fi
\ifx \shownote     \undefined \def \shownote      #1{#1}          \fi
\ifx \showarticletitle \undefined \def \showarticletitle #1{#1}   \fi
\ifx \showURL      \undefined \def \showURL       {\relax}        \fi
\providecommand\bibfield[2]{#2}
\providecommand\bibinfo[2]{#2}
\providecommand\natexlab[1]{#1}
\providecommand\showeprint[2][]{arXiv:#2}

\bibitem[\protect\citeauthoryear{Akiyama, Avis, Chv{\'a}tal, and Era}{Akiyama
  et~al\mbox{.}}{1981}]%
        {akiyama1981balancing}
\bibfield{author}{\bibinfo{person}{Jin Akiyama}, \bibinfo{person}{David Avis},
  \bibinfo{person}{Vasek Chv{\'a}tal}, {and} \bibinfo{person}{Hiroshi Era}.}
  \bibinfo{year}{1981}\natexlab{}.
\newblock \showarticletitle{Balancing signed graphs}.
\newblock \bibinfo{journal}{\emph{Discrete Applied Mathematics}}
  \bibinfo{volume}{3}, \bibinfo{number}{4} (\bibinfo{year}{1981}),
  \bibinfo{pages}{227--233}.
\newblock


\bibitem[\protect\citeauthoryear{Anchuri and Magdon-Ismail}{Anchuri and
  Magdon-Ismail}{2012}]%
        {anchuri2012communities}
\bibfield{author}{\bibinfo{person}{Pranay Anchuri} {and} \bibinfo{person}{Malik
  Magdon-Ismail}.} \bibinfo{year}{2012}\natexlab{}.
\newblock \showarticletitle{Communities and balance in signed networks: A
  spectral approach}. In \bibinfo{booktitle}{\emph{Proceedings of the 2012
  International Conference on Advances in Social Networks Analysis and Mining
  (ASONAM 2012)}}. IEEE Computer Society, \bibinfo{pages}{235--242}.
\newblock


\bibitem[\protect\citeauthoryear{Arbenz and Golub}{Arbenz and Golub}{1988}]%
        {arbenz1988spectral}
\bibfield{author}{\bibinfo{person}{Peter Arbenz} {and} \bibinfo{person}{Gene~H
  Golub}.} \bibinfo{year}{1988}\natexlab{}.
\newblock \showarticletitle{On the spectral decomposition of Hermitian matrices
  modified by low rank perturbations with applications}.
\newblock \bibinfo{journal}{\emph{SIAM J. Matrix Anal. Appl.}}
  \bibinfo{volume}{9}, \bibinfo{number}{1} (\bibinfo{year}{1988}),
  \bibinfo{pages}{40--58}.
\newblock


\bibitem[\protect\citeauthoryear{Arora, Rao, and Vazirani}{Arora
  et~al\mbox{.}}{2009}]%
        {arora2009expander}
\bibfield{author}{\bibinfo{person}{Sanjeev Arora}, \bibinfo{person}{Satish
  Rao}, {and} \bibinfo{person}{Umesh Vazirani}.}
  \bibinfo{year}{2009}\natexlab{}.
\newblock \showarticletitle{Expander flows, geometric embeddings and graph
  partitioning}.
\newblock \bibinfo{journal}{\emph{Journal of the ACM (JACM)}}
  \bibinfo{volume}{56}, \bibinfo{number}{2} (\bibinfo{year}{2009}),
  \bibinfo{pages}{5}.
\newblock


\bibitem[\protect\citeauthoryear{Barab{\'a}si and Albert}{Barab{\'a}si and
  Albert}{1999}]%
        {barabasi1999emergence}
\bibfield{author}{\bibinfo{person}{Albert-L{\'a}szl{\'o} Barab{\'a}si} {and}
  \bibinfo{person}{R{\'e}ka Albert}.} \bibinfo{year}{1999}\natexlab{}.
\newblock \showarticletitle{Emergence of scaling in random networks}.
\newblock \bibinfo{journal}{\emph{Acience}} \bibinfo{volume}{286},
  \bibinfo{number}{5439} (\bibinfo{year}{1999}), \bibinfo{pages}{509--512}.
\newblock


\bibitem[\protect\citeauthoryear{Bonchi, Galimberti, Gionis, Ordozgoiti, and
  Ruffo}{Bonchi et~al\mbox{.}}{2019}]%
        {bonchi2019discovering}
\bibfield{author}{\bibinfo{person}{Francesco Bonchi}, \bibinfo{person}{Edoardo
  Galimberti}, \bibinfo{person}{Aristides Gionis}, \bibinfo{person}{Bruno
  Ordozgoiti}, {and} \bibinfo{person}{Giancarlo Ruffo}.}
  \bibinfo{year}{2019}\natexlab{}.
\newblock \showarticletitle{Discovering polarized communities in signed
  networks}. In \bibinfo{booktitle}{\emph{Proceedings of the 28th ACM
  International Conference on Information and Knowledge Management}}.
  \bibinfo{pages}{961--970}.
\newblock


\bibitem[\protect\citeauthoryear{Cartwright and Harary}{Cartwright and
  Harary}{1956}]%
        {cartwright1956structural}
\bibfield{author}{\bibinfo{person}{Dorwin Cartwright} {and}
  \bibinfo{person}{Frank Harary}.} \bibinfo{year}{1956}\natexlab{}.
\newblock \showarticletitle{Structural balance: a generalization of Heider's
  theory}.
\newblock \bibinfo{journal}{\emph{Psychological review}} \bibinfo{volume}{63},
  \bibinfo{number}{5} (\bibinfo{year}{1956}), \bibinfo{pages}{277}.
\newblock


\bibitem[\protect\citeauthoryear{Chu, Wang, Pei, Wang, Zhao, and Chen}{Chu
  et~al\mbox{.}}{2016}]%
        {chu2016finding}
\bibfield{author}{\bibinfo{person}{Lingyang Chu}, \bibinfo{person}{Zhefeng
  Wang}, \bibinfo{person}{Jian Pei}, \bibinfo{person}{Jiannan Wang},
  \bibinfo{person}{Zijin Zhao}, {and} \bibinfo{person}{Enhong Chen}.}
  \bibinfo{year}{2016}\natexlab{}.
\newblock \showarticletitle{Finding gangs in war from signed networks}. In
  \bibinfo{booktitle}{\emph{Proceedings of the 22nd ACM SIGKDD International
  Conference on Knowledge Discovery and Data Mining}}. ACM,
  \bibinfo{pages}{1505--1514}.
\newblock


\bibitem[\protect\citeauthoryear{Crowston, Gutin, Jones, and
  Muciaccia}{Crowston et~al\mbox{.}}{2013}]%
        {crowston2013maximum}
\bibfield{author}{\bibinfo{person}{Robert Crowston}, \bibinfo{person}{Gregory
  Gutin}, \bibinfo{person}{Mark Jones}, {and} \bibinfo{person}{Gabriele
  Muciaccia}.} \bibinfo{year}{2013}\natexlab{}.
\newblock \showarticletitle{Maximum balanced subgraph problem parameterized
  above lower bound}.
\newblock \bibinfo{journal}{\emph{Theoretical Computer Science}}
  \bibinfo{volume}{513} (\bibinfo{year}{2013}), \bibinfo{pages}{53--64}.
\newblock


\bibitem[\protect\citeauthoryear{DasGupta, Enciso, Sontag, and Zhang}{DasGupta
  et~al\mbox{.}}{2007}]%
        {dasgupta2007algorithmic}
\bibfield{author}{\bibinfo{person}{Bhaskar DasGupta},
  \bibinfo{person}{German~Andres Enciso}, \bibinfo{person}{Eduardo Sontag},
  {and} \bibinfo{person}{Yi Zhang}.} \bibinfo{year}{2007}\natexlab{}.
\newblock \showarticletitle{Algorithmic and complexity results for
  decompositions of biological networks into monotone subsystems}.
\newblock \bibinfo{journal}{\emph{Biosystems}} \bibinfo{volume}{90},
  \bibinfo{number}{1} (\bibinfo{year}{2007}), \bibinfo{pages}{161--178}.
\newblock


\bibitem[\protect\citeauthoryear{Doreian and Mrvar}{Doreian and Mrvar}{1996}]%
        {doreian1996partitioning}
\bibfield{author}{\bibinfo{person}{Patrick Doreian} {and}
  \bibinfo{person}{Andrej Mrvar}.} \bibinfo{year}{1996}\natexlab{}.
\newblock \showarticletitle{A partitioning approach to structural balance}.
\newblock \bibinfo{journal}{\emph{Social networks}} \bibinfo{volume}{18},
  \bibinfo{number}{2} (\bibinfo{year}{1996}), \bibinfo{pages}{149--168}.
\newblock


\bibitem[\protect\citeauthoryear{Figueiredo and Frota}{Figueiredo and
  Frota}{2014}]%
        {figueiredo2014maximum}
\bibfield{author}{\bibinfo{person}{Rosa Figueiredo} {and} \bibinfo{person}{Yuri
  Frota}.} \bibinfo{year}{2014}\natexlab{}.
\newblock \showarticletitle{The maximum balanced subgraph of a signed graph:
  Applications and solution approaches}.
\newblock \bibinfo{journal}{\emph{European Journal of Operational Research}}
  \bibinfo{volume}{236}, \bibinfo{number}{2} (\bibinfo{year}{2014}),
  \bibinfo{pages}{473--487}.
\newblock


\bibitem[\protect\citeauthoryear{Fortunato}{Fortunato}{2010}]%
        {fortunato2010community}
\bibfield{author}{\bibinfo{person}{Santo Fortunato}.}
  \bibinfo{year}{2010}\natexlab{}.
\newblock \showarticletitle{Community detection in graphs}.
\newblock \bibinfo{journal}{\emph{Physics reports}} \bibinfo{volume}{486},
  \bibinfo{number}{3-5} (\bibinfo{year}{2010}), \bibinfo{pages}{75--174}.
\newblock


\bibitem[\protect\citeauthoryear{Garimella, De~Francisci~Morales, Gionis, and
  Mathioudakis}{Garimella et~al\mbox{.}}{2017}]%
        {garimella2017reducing}
\bibfield{author}{\bibinfo{person}{Kiran Garimella}, \bibinfo{person}{Gianmarco
  De~Francisci~Morales}, \bibinfo{person}{Aristides Gionis}, {and}
  \bibinfo{person}{Michael Mathioudakis}.} \bibinfo{year}{2017}\natexlab{}.
\newblock \showarticletitle{Reducing controversy by connecting opposing views}.
  In \bibinfo{booktitle}{\emph{Proceedings of the Tenth ACM International
  Conference on Web Search and Data Mining}}. ACM, \bibinfo{pages}{81--90}.
\newblock


\bibitem[\protect\citeauthoryear{Garimella, Morales, Gionis, and
  Mathioudakis}{Garimella et~al\mbox{.}}{2018}]%
        {garimella2018quantifying}
\bibfield{author}{\bibinfo{person}{Kiran Garimella}, \bibinfo{person}{Gianmarco
  De~Francisci Morales}, \bibinfo{person}{Aristides Gionis}, {and}
  \bibinfo{person}{Michael Mathioudakis}.} \bibinfo{year}{2018}\natexlab{}.
\newblock \showarticletitle{Quantifying controversy on social media}.
\newblock \bibinfo{journal}{\emph{ACM Transactions on Social Computing}}
  \bibinfo{volume}{1}, \bibinfo{number}{1} (\bibinfo{year}{2018}),
  \bibinfo{pages}{3}.
\newblock


\bibitem[\protect\citeauthoryear{G{\"u}lpinar, Gutin, Mitra, and
  Zverovitch}{G{\"u}lpinar et~al\mbox{.}}{2004}]%
        {gulpinar2004extracting}
\bibfield{author}{\bibinfo{person}{Nal{\^a}n G{\"u}lpinar},
  \bibinfo{person}{Gregory Gutin}, \bibinfo{person}{Gautam Mitra}, {and}
  \bibinfo{person}{Alexey Zverovitch}.} \bibinfo{year}{2004}\natexlab{}.
\newblock \showarticletitle{Extracting pure network submatrices in linear
  programs using signed graphs}.
\newblock \bibinfo{journal}{\emph{Discrete Applied Mathematics}}
  \bibinfo{volume}{137}, \bibinfo{number}{3} (\bibinfo{year}{2004}),
  \bibinfo{pages}{359--372}.
\newblock


\bibitem[\protect\citeauthoryear{Harary}{Harary}{1953}]%
        {harary1953notion}
\bibfield{author}{\bibinfo{person}{Frank Harary}.}
  \bibinfo{year}{1953}\natexlab{}.
\newblock \showarticletitle{On the notion of balance of a signed graph}.
\newblock \bibinfo{journal}{\emph{The Michigan Mathematical Journal}}
  \bibinfo{volume}{2}, \bibinfo{number}{2} (\bibinfo{year}{1953}),
  \bibinfo{pages}{143--146}.
\newblock


\bibitem[\protect\citeauthoryear{Harary and Kabell}{Harary and Kabell}{1980}]%
        {harary1980simple}
\bibfield{author}{\bibinfo{person}{Frank Harary} {and}
  \bibinfo{person}{Jerald~A Kabell}.} \bibinfo{year}{1980}\natexlab{}.
\newblock \showarticletitle{A simple algorithm to detect balance in signed
  graphs}.
\newblock \bibinfo{journal}{\emph{Mathematical Social Sciences}}
  \bibinfo{volume}{1}, \bibinfo{number}{1} (\bibinfo{year}{1980}),
  \bibinfo{pages}{131--136}.
\newblock


\bibitem[\protect\citeauthoryear{Hou, Li, and Pan}{Hou et~al\mbox{.}}{2003}]%
        {hou2003laplacian}
\bibfield{author}{\bibinfo{person}{Yaoping Hou}, \bibinfo{person}{Jiongsheng
  Li}, {and} \bibinfo{person}{Yongliang Pan}.} \bibinfo{year}{2003}\natexlab{}.
\newblock \showarticletitle{On the {L}aplacian eigenvalues of signed graphs}.
\newblock \bibinfo{journal}{\emph{Linear and Multilinear Algebra}}
  \bibinfo{volume}{51}, \bibinfo{number}{1} (\bibinfo{year}{2003}),
  \bibinfo{pages}{21--30}.
\newblock


\bibitem[\protect\citeauthoryear{Hou}{Hou}{2005}]%
        {hou2005bounds}
\bibfield{author}{\bibinfo{person}{Yao~Ping Hou}.}
  \bibinfo{year}{2005}\natexlab{}.
\newblock \showarticletitle{Bounds for the least {L}aplacian eigenvalue of a
  signed graph}.
\newblock \bibinfo{journal}{\emph{Acta Mathematica Sinica}}
  \bibinfo{volume}{21}, \bibinfo{number}{4} (\bibinfo{year}{2005}),
  \bibinfo{pages}{955--960}.
\newblock


\bibitem[\protect\citeauthoryear{H{\"u}ffner, Betzler, and
  Niedermeier}{H{\"u}ffner et~al\mbox{.}}{2007}]%
        {huffner2007optimal}
\bibfield{author}{\bibinfo{person}{Falk H{\"u}ffner}, \bibinfo{person}{Nadja
  Betzler}, {and} \bibinfo{person}{Rolf Niedermeier}.}
  \bibinfo{year}{2007}\natexlab{}.
\newblock \showarticletitle{Optimal edge deletions for signed graph balancing}.
  In \bibinfo{booktitle}{\emph{International Workshop on Experimental and
  Efficient Algorithms}}. Springer, \bibinfo{pages}{297--310}.
\newblock


\bibitem[\protect\citeauthoryear{Knyazev}{Knyazev}{2001}]%
        {knyazev2001toward}
\bibfield{author}{\bibinfo{person}{Andrew~V Knyazev}.}
  \bibinfo{year}{2001}\natexlab{}.
\newblock \showarticletitle{Toward the optimal preconditioned eigensolver:
  Locally optimal block preconditioned conjugate gradient method}.
\newblock \bibinfo{journal}{\emph{SIAM journal on scientific computing}}
  \bibinfo{volume}{23}, \bibinfo{number}{2} (\bibinfo{year}{2001}),
  \bibinfo{pages}{517--541}.
\newblock


\bibitem[\protect\citeauthoryear{Kunegis, Schmidt, Lommatzsch, Lerner, De~Luca,
  and Albayrak}{Kunegis et~al\mbox{.}}{2010}]%
        {kunegis2010spectral}
\bibfield{author}{\bibinfo{person}{J{\'e}r{\^o}me Kunegis},
  \bibinfo{person}{Stephan Schmidt}, \bibinfo{person}{Andreas Lommatzsch},
  \bibinfo{person}{J{\"u}rgen Lerner}, \bibinfo{person}{Ernesto~W De~Luca},
  {and} \bibinfo{person}{Sahin Albayrak}.} \bibinfo{year}{2010}\natexlab{}.
\newblock \showarticletitle{Spectral analysis of signed graphs for clustering,
  prediction and visualization}. In \bibinfo{booktitle}{\emph{Proceedings of
  the 2010 SIAM International Conference on Data Mining}}. SIAM,
  \bibinfo{pages}{559--570}.
\newblock


\bibitem[\protect\citeauthoryear{Lai, Patti, Ruffo, and Rosso}{Lai
  et~al\mbox{.}}{2018}]%
        {lai2018stance}
\bibfield{author}{\bibinfo{person}{Mirko Lai}, \bibinfo{person}{Viviana Patti},
  \bibinfo{person}{Giancarlo Ruffo}, {and} \bibinfo{person}{Paolo Rosso}.}
  \bibinfo{year}{2018}\natexlab{}.
\newblock \showarticletitle{Stance Evolution and Twitter Interactions in an
  Italian Political Debate}. In \bibinfo{booktitle}{\emph{International
  Conference on Applications of Natural Language to Information Systems}}.
  Springer, \bibinfo{pages}{15--27}.
\newblock


\bibitem[\protect\citeauthoryear{Li and Li}{Li and Li}{2016}]%
        {li2016note}
\bibfield{author}{\bibinfo{person}{Hui~Shu Li} {and} \bibinfo{person}{Hong~Hai
  Li}.} \bibinfo{year}{2016}\natexlab{}.
\newblock \showarticletitle{A note on the least (normalized) laplacian
  eigenvalue of signed graphs}.
\newblock \bibinfo{journal}{\emph{Tamkang Journal of Mathematics}}
  \bibinfo{volume}{47}, \bibinfo{number}{3} (\bibinfo{year}{2016}),
  \bibinfo{pages}{271--278}.
\newblock


\bibitem[\protect\citeauthoryear{Liao and Fu}{Liao and Fu}{2014}]%
        {liao2014can}
\bibfield{author}{\bibinfo{person}{Q~Vera Liao} {and} \bibinfo{person}{Wai-Tat
  Fu}.} \bibinfo{year}{2014}\natexlab{}.
\newblock \showarticletitle{Can you hear me now?: mitigating the echo chamber
  effect by source position indicators}. In
  \bibinfo{booktitle}{\emph{Proceedings of the 17th ACM conference on Computer
  supported cooperative work \& social computing}}. ACM,
  \bibinfo{pages}{184--196}.
\newblock


\bibitem[\protect\citeauthoryear{Lo, Surian, Zhang, and Lim}{Lo
  et~al\mbox{.}}{2011}]%
        {lo2011mining}
\bibfield{author}{\bibinfo{person}{David Lo}, \bibinfo{person}{Didi Surian},
  \bibinfo{person}{Kuan Zhang}, {and} \bibinfo{person}{Ee-Peng Lim}.}
  \bibinfo{year}{2011}\natexlab{}.
\newblock \showarticletitle{Mining direct antagonistic communities in explicit
  trust networks}. In \bibinfo{booktitle}{\emph{Proceedings of the 20th ACM
  international conference on Information and knowledge management}}. ACM,
  \bibinfo{pages}{1013--1018}.
\newblock


\bibitem[\protect\citeauthoryear{Mejova, Zhang, Diakopoulos, and
  Castillo}{Mejova et~al\mbox{.}}{2014}]%
        {mejova2014controversy}
\bibfield{author}{\bibinfo{person}{Yelena Mejova}, \bibinfo{person}{Amy~X
  Zhang}, \bibinfo{person}{Nicholas Diakopoulos}, {and} \bibinfo{person}{Carlos
  Castillo}.} \bibinfo{year}{2014}\natexlab{}.
\newblock \showarticletitle{Controversy and sentiment in online news}.
\newblock \bibinfo{journal}{\emph{arXiv preprint arXiv:1409.8152}}
  (\bibinfo{year}{2014}).
\newblock


\bibitem[\protect\citeauthoryear{Morales, Borondo, Losada, and Benito}{Morales
  et~al\mbox{.}}{2015}]%
        {morales2015measuring}
\bibfield{author}{\bibinfo{person}{AJ Morales}, \bibinfo{person}{Javier
  Borondo}, \bibinfo{person}{Juan~Carlos Losada}, {and} \bibinfo{person}{Rosa~M
  Benito}.} \bibinfo{year}{2015}\natexlab{}.
\newblock \showarticletitle{Measuring political polarization: Twitter shows the
  two sides of Venezuela}.
\newblock \bibinfo{journal}{\emph{Chaos: An Interdisciplinary Journal of
  Nonlinear Science}} \bibinfo{volume}{25}, \bibinfo{number}{3}
  (\bibinfo{year}{2015}), \bibinfo{pages}{033114}.
\newblock


\bibitem[\protect\citeauthoryear{Page, Brin, Motwani, and Winograd}{Page
  et~al\mbox{.}}{1999}]%
        {page1999pagerank}
\bibfield{author}{\bibinfo{person}{Lawrence Page}, \bibinfo{person}{Sergey
  Brin}, \bibinfo{person}{Rajeev Motwani}, {and} \bibinfo{person}{Terry
  Winograd}.} \bibinfo{year}{1999}\natexlab{}.
\newblock \bibinfo{booktitle}{\emph{The PageRank citation ranking: Bringing
  order to the web.}}
\newblock \bibinfo{type}{{T}echnical {R}eport}. \bibinfo{institution}{Stanford
  InfoLab}.
\newblock


\bibitem[\protect\citeauthoryear{Poljak and Turz{\'\i}k}{Poljak and
  Turz{\'\i}k}{1986}]%
        {poljak1986polynomial}
\bibfield{author}{\bibinfo{person}{Svatopluk Poljak} {and}
  \bibinfo{person}{Daniel Turz{\'\i}k}.} \bibinfo{year}{1986}\natexlab{}.
\newblock \showarticletitle{A polynomial time heuristic for certain subgraph
  optimization problems with guaranteed worst case bound}.
\newblock \bibinfo{journal}{\emph{Discrete Mathematics}} \bibinfo{volume}{58},
  \bibinfo{number}{1} (\bibinfo{year}{1986}), \bibinfo{pages}{99--104}.
\newblock


\bibitem[\protect\citeauthoryear{Tang, Chang, Aggarwal, and Liu}{Tang
  et~al\mbox{.}}{2016}]%
        {tang2016survey}
\bibfield{author}{\bibinfo{person}{Jiliang Tang}, \bibinfo{person}{Yi Chang},
  \bibinfo{person}{Charu Aggarwal}, {and} \bibinfo{person}{Huan Liu}.}
  \bibinfo{year}{2016}\natexlab{}.
\newblock \showarticletitle{A survey of signed network mining in social media}.
\newblock \bibinfo{journal}{\emph{ACM Computing Surveys (CSUR)}}
  \bibinfo{volume}{49}, \bibinfo{number}{3} (\bibinfo{year}{2016}),
  \bibinfo{pages}{42}.
\newblock


\bibitem[\protect\citeauthoryear{Vydiswaran, Zhai, Roth, and
  Pirolli}{Vydiswaran et~al\mbox{.}}{2015}]%
        {vydiswaran2015overcoming}
\bibfield{author}{\bibinfo{person}{VG~Vinod Vydiswaran},
  \bibinfo{person}{ChengXiang Zhai}, \bibinfo{person}{Dan Roth}, {and}
  \bibinfo{person}{Peter Pirolli}.} \bibinfo{year}{2015}\natexlab{}.
\newblock \showarticletitle{Overcoming bias to learn about controversial
  topics}.
\newblock \bibinfo{journal}{\emph{Journal of the Association for Information
  Science and Technology}} \bibinfo{volume}{66}, \bibinfo{number}{8}
  (\bibinfo{year}{2015}), \bibinfo{pages}{1655--1672}.
\newblock


\bibitem[\protect\citeauthoryear{Yang, Cheung, and Liu}{Yang
  et~al\mbox{.}}{2007}]%
        {yang2007community}
\bibfield{author}{\bibinfo{person}{Bo Yang}, \bibinfo{person}{William Cheung},
  {and} \bibinfo{person}{Jiming Liu}.} \bibinfo{year}{2007}\natexlab{}.
\newblock \showarticletitle{Community mining from signed social networks}.
\newblock \bibinfo{journal}{\emph{IEEE Transactions on Knowledge and Data
  Engineering}} \bibinfo{volume}{19}, \bibinfo{number}{10}
  (\bibinfo{year}{2007}), \bibinfo{pages}{1333--1348}.
\newblock


\bibitem[\protect\citeauthoryear{Zaslavsky}{Zaslavsky}{1982}]%
        {zaslavsky1982signed}
\bibfield{author}{\bibinfo{person}{Thomas Zaslavsky}.}
  \bibinfo{year}{1982}\natexlab{}.
\newblock \showarticletitle{Signed graphs}.
\newblock \bibinfo{journal}{\emph{Discrete Applied Mathematics}}
  \bibinfo{volume}{4}, \bibinfo{number}{1} (\bibinfo{year}{1982}),
  \bibinfo{pages}{47--74}.
\newblock


\bibitem[\protect\citeauthoryear{Zaslavsky}{Zaslavsky}{2012}]%
        {zaslavsky2012mathematical}
\bibfield{author}{\bibinfo{person}{Thomas Zaslavsky}.}
  \bibinfo{year}{2012}\natexlab{}.
\newblock \showarticletitle{A mathematical bibliography of signed and gain
  graphs and allied areas}.
\newblock \bibinfo{journal}{\emph{The Electronic Journal of Combinatorics}}
  \bibinfo{volume}{1000} (\bibinfo{year}{2012}), \bibinfo{pages}{8--6}.
\newblock


\end{thebibliography}

\end{document}